\documentclass[3p,times]{elsarticle}
\usepackage{ecrc}
\volume{00}
\firstpage{1}
\journalname{Engineering Applications of Artificial Intelligence}
\runauth{M. Smyrnakis and S. Veres}
\jid{eaai}
\usepackage{amssymb}
\usepackage{amsthm}
\usepackage{amsmath}
\usepackage{amsfonts}
\usepackage[figuresright]{rotating}
\usepackage{multirow}
\usepackage{subfigure}
\usepackage{inputenc}
\usepackage{float}
\usepackage{color}
\newtheorem{thm}{Theorem}
\newtheorem{prop}{Proposition}

\begin{document}

\begin{frontmatter}

\title{Fictitious play for cooperative action selection in robot teams}

\author[a1]{M. Smyrnakis\corref{cor1}}
\ead{m.smyrnakis@sheffield.ac.uk}
\author[a1]{S. Veres}
\ead{s.veres@sheffield.ac.uk}
\address[a1]{Department of Automatic Control and Systems Engineering
The University of Sheffield
Mappin Street
Sheffield, S1 3JD
United Kingdom }

\begin{abstract}
A game theoretic distributed decision making approach is presented for the problem of control effort allocation in a robotic team  based on a novel variant of fictitious play. The proposed learning process allows the robots to accomplish their objectives by coordinating their actions in order to efficiently complete their tasks. In particular, each robot of the team  predicts the other robots' planned actions while making decisions to maximise their own expected reward that depends on the reward for joint  successful completion of the task.  Action selection is interpreted as an $n$-player cooperative game. The approach presented  can be seen as part of the \emph{Belief Desire Intention} (BDI) framework, also can address  the problem of  cooperative, legal, safe, considerate and emphatic decisions by robots if their individual and group rewards are suitably defined. After theoretical analysis the performance of the proposed algorithm is tested on four simulation scenarios. The first one is a coordination game between two material handling robots,  the second one is a warehouse patrolling task by a team of robots, the third one presents a coordination mechanism between two robots that carry a heavy object on a corridor and the fourth one is an example of coordination on a sensors network. 
\end{abstract}

\begin{keyword}
Robot team coordination; Fictitious play; Extended Kalman filters; Game theory; Distributed optimisation


\end{keyword}

\end{frontmatter}

\section{Introduction}
\label{intro}
Recent advances in industrial automation technology often require distributed  optimisation in a multi-agent system where each agent controls a machine. 
An application of particular interest,  addressed in this paper through a game-theoretic approach, is the coordination of robot teams. Teams of robots can be used in many domains such as mine detection \citep{minedetect}, medication delivery in medical facilities \citep{med_fac}, formation control \cite{flight} and exploration of unknown environments \cite{exploration,env1,env2}. In these cases teams of intelligent robots should coordinate in order to accomplish a desired task. When autonomy is a desired property of a multi-robot system then self-coordination is necessary between the robots of the team. Applications of these methodologies also include wireless sensor networks~\cite{wide_area,sn,animal_tracking,monitoring}, smart grids~\cite{smart_grid1,smart_grid2}, water distribution system optimisation \cite{water} and scheduling problems \cite{sc}. 

Game theory has also been used to design optimal controllers when the objective is coordination, see e.g. \cite{game_robot1}.  Using this approach   the agents/robots  eventually reach the Nash equilibrium of a coordination game. In \cite{game_robot1,game_robot2} local and global components in the agents' cost function are used and the Nash equilibrium of the game is reached. Another approach is presented in \cite{game_robot3}, based on agents' cost functions, which use local components and the assumption that the states of the other agents are constant.

Fictitious play is an iterative learning process where players choose an action that maximises their expected rewards based on their beliefs about their opponents' strategies. The players update these beliefs after observing their opponents' actions. Even though fictitious play converges to the Nash equilibrium for certain categories of games \citep{fp1,fp2,fp3,fp4,learning_in_games}, this convergence can be very slow because of the assumption that players use a fixed strategy in the whole game \cite{learning_in_games}. Speed up of the convergence can be facilitated by an alternative approach, which was presented in \cite{cj}, where opponents' strategies vary through time and players  use particle filters to predict them. Though providing faster convergence, this approach has the drawback of high computational costs of the particle filters. In applications where the computational cost is important, as the coordination of many UAVs, the particle filters approach is intractable. An alternative, we propose here, is to use  extended Kalman filters (EKF) instead of  predicting opponents' strategies using particle filters. EKFs have much smaller computational costs than the particle filter variant of fictitious play algorithm that has been proposed in \cite{cj}. Moreover in contrast to \citep{cj} we provide a proof of convergence to Nash equilibrium of the proposed learning algorithm for potential games. Potential games are of particular interest as many distributed optimisation tasks can be cast as potential games. Therefore convergence of an algorithm to the the Nash equilibrium of a potential game is equivalent to convergence to the global or local optimum of the distributed optimisation problem.

Thus the proposed learning process can be used as a design methodology for cooperative control based on game theory, which overlaps with solutions in the area of distributed optimisation \citep{dist1}. Each agent $i$ strives to maximise a global control reward, as negative of the control cost, through minimising its private control cost, which is associated with the global one. The private cost function of an agent $i$ incorporates terms that not only depend on agent $i$, but also on costs associated with the actions of other agents. As the agents strive to minimise a common  cost function through their individual ones, the problem addressed here can also be seen as a distributed optimisation problem. In this work we enable the agents to learn how they minimise their cost function through communication and interaction with other agents, instead of finding the Nash equilibrium of the game, which is not possible in polynomial time for some games \citep{daskalakis}. In the proposed scheme robots learn and change their behaviour according to the other robots' actions. The learning algorithm, which is based on fictitious play \citep{brown_fict}, serves as the coordination mechanism of the controllers of team members. Thus, in the proposed cooperative control methodology there is an implicit coordination phase where agents learn other agents' policies and then they use this knowledge to decide on the action that minimises their cost functions. Additionally the proposed control module can be seen as a part of the BDI framework since each agent updates his beliefs about his opponents' strategies given the state of the environment and based on a decision rule that can represent his desires perform the selected actions



The remainder of this paper is organised as follows. We start with a brief description of relevant game-theoretic definitions. Section \ref{sect:rac} present some background material about Rational agents.  
Section \ref{new_algorithm} introduces the learning algorithm that we use in our controller, Section \ref{theory} contains the main theoretical results and Section \ref{parameters} contains the simulation results in order to define the parameters of the proposed algorithm.   Section \ref{simulation}  presents   simulation results  before  conclusions are drawn. 

\section{Game theoretical definitions}
\label{sec:games}
In this section we will briefly present some basic definitions from game theory, since the learning block of our controller is based on these. A game $\Gamma$ is defined by a set of players $\mathcal{I}$, $i \in \{1,2,\ldots,\mathcal{I}\}$, who can choose an action, $s^{i}$, from a finite discrete set $S^{i}$. We then can define the joint action $s$, $s=(s^{1}, \ldots, s^{\mathcal{I}})$, that is played in a game as an element of the  product set \mbox{$S=\times_{i=1}^{i=\mathcal{I}}S^{i}$}. Each player $i$ receives a reward, $r^{i}$, after choosing an action $s^{i}$. The reward, also called the utility, is a map from the joint action space to the real numbers, $r^{i}:S \rightarrow R$. We will often write $s=(s^{i},s^{-i})$, where $s^{i}$ is the action of player $i$ and $s^{-i}$ is the joint action of player $i$'s opponents. When players select their actions using a probability distribution they use mixed strategies. The mixed strategy of a player $i$, $\sigma^{i}$, is an element of the set $\Delta^{i}$, where $\Delta^{i}$ is the set of all the probability distributions over the action space $S^{i}$. The joint mixed strategy, $\sigma$, is then an element of $\Delta=\times_{i=1}^{i=\mathbb{I}}\Delta^{i}$. A strategy where a specific action is chosen with probability 1 is referred to as pure strategy. Analogously to the joint actions we will write  $\sigma=(\sigma^{i},\sigma^{-i})$ for mixed strategies. The expected utility a player $i$ will gain if it chooses a strategy $\sigma^{i}$ (resp.\ $s^{i}$), when its opponents choose the joint strategy $\sigma^{-i}$,  is denoted by $r^{i}(\sigma^{i},\sigma^{-i})$ (resp.\ $r^{i}(s^{i},\sigma^{-i})$).

A game, depending on the structure of its reward functions, can be characterised either as competitive or as a coordination game. In competitive games players have conflicted interest and there is not a single joint action where all players maximise their utilities. Zero sum games are a representative example of competitive games where the reward of a player $i$   is the loss of other players . An example of a zero-sum game is presented in Table \ref{tab:matching}. 
\begin{table}%
\centering
\begin{tabular}{|l|c|r|}
\hline
&Head &Tails \\ \hline
Head & 1,-1&-1,1 \\ \hline
Tails & -1,1&1,-1\\ \hline
\end{tabular}
\caption{Rewards of two players in a zero sum game as function of the outcome of throwing a coin: matching pennies game.  }
\label{tab:matching}
\end{table}
On the other hand in coordination games players either share a common reward function or their rewards are maximised in the same joint action. A very simple example of a coordination game where players share the same rewards, is depicted in Table \ref{tab:simcoord}.
\begin{table}
\centering
\begin{tabular}{|c| c| c|}
\hline
 &L&R\\ \hline
 U& 1,1 & 0,0 \\ \hline
 D& 0,0 & 1,1 \\ \hline
 \end{tabular}
\caption{Rewards of two players  in a simple coordination game as function of joint moves to the left \& up (L,U) or right \& down (R).}
 \label{tab:simcoord}
\end{table}
Even though competitive games are the most studied games, we will focus our work on coordination games because they naturally formulate a solution to distributed optimisation and coordination. 

\subsection{Best response and Nash Equilibrium}
A common decision rule in game theory is best response. Best response is defined as the action that maximises players' expected utility given their opponents' strategies. Thus for a specific mixed strategy  $\sigma^{-i}$ we evaluate the best response as:
\begin{equation}
\hat{\sigma}^{i}_{pure}(\sigma^{-i})= \mathop{\rm argmax}_{s^{i} \in S} \quad
r^{i}(s^{i},\sigma^{-i})
\label{eq:br}
\end{equation}
Nash in \cite{nash} showed that every game has at least one equilibrium. A joint mixed strategy $\hat{\sigma}=(\hat{\sigma}^{i},\hat{\sigma}^{-i})$ is called a Nash equilibrium when
\begin{equation}
r^{i}(\hat{\sigma}^{i},\hat{\sigma}^{-i})\geq r^{i}(\sigma^{i},\hat{\sigma}^{-i}) \qquad \textrm{for any }  \sigma^{i} \in \Delta^{i}
\label{eq:nashutil}
\end{equation}
Equation (\ref{eq:nashutil}) implies that if a strategy $\hat{\sigma}$ is a Nash equilibrium then it is not possible for a player to increase its utility by unilaterally changing its strategy. When all the robotic players in a game select their actions using pure strategies then the equilibrium is referred as pure Nash equilibrium.

\subsection{Optimisation tasks as potential games}
It is possible to cast distributed optimisation tasks as potential games \citep{autonomous,archie}, thus the task of finding an optimal solution of the distributed optimisation task can be seen as the search of a Nash equilibrium in a game. An optimisation problem can be solved distributively if it can be divided into $\mathcal{D}$ coupled or independent sub-problems with the following property \cite{berts}:
\begin{equation}
r (s)-r (\tilde{s})>0  \Leftrightarrow
r^{i}(s^i)-r^{i}(\tilde{s}^{i})>0,\ i \in \mathcal{I}, \forall s, \tilde{s}
\label{eq:masprop}
\end{equation}
where $s$ and $\tilde{s}$ are any sets of actions by the agents, $r$ represents the global reward function or global utility, and $r^{i},\  i \in \mathcal{I}$, represent    players' reward or utility. Equation (\ref{eq:masprop}) implies that  a joint action $s$, should have the similarly positive or negative impact   in the global and the local task. Thus a solution $s$ should increase or decrease both the local and global utility when it is compared with another joint action $\tilde{s}$.

There is a direct analogy between (\ref{eq:masprop}) and ordinal potential games. Ordinal potential games are games where their reward function have a potential function with the following property \citep{fp4}
\begin{equation}
r^{i}(s^{i},s^{-i})-r^{i}(\tilde{s}^{i},\tilde{s}^{-i})>0 \Leftrightarrow
\phi(s^{i},s^{-i})-\phi(\tilde{s}^{i},\tilde{s}^{-i})>0, \forall s=(s^{i},s^{-i}),\ \forall \tilde{s}=(\tilde{s}^{i},\tilde{s}^{-i}) \label{eq:poten}
\end{equation}
where $\phi$ is a potential function and the above equality stands for every player $i$. Exact potential games (or potential games thereafter), is a subclass of ordinal potential games which can be used to solve distributed optimisation problems where the difference in the global reward between two  joint actions is the same as the difference in the potential function \citep{fp4}:
\begin{equation}
r^{i}(s^{i},s^{-i})-r^{i}(\tilde{s}^{i},\tilde{s}^{-i})=
\phi(s^{i},s^{-i})-\phi(\tilde{s}^{i},\tilde{s}^{-i}),\ \forall s=(s^{i},s^{-i}),\ \forall \tilde{s}=(\tilde{s}^{i},\tilde{s}^{-i})
\label{eq:pot}
\end{equation}
where similarly to ordinal potential games $\phi$ is a potential function and the above equality stands for every player $i$.
An advantage of potential games is that they have at least one pure Nash equilibrium, hence there is at least one joint action $s$ where no player can increase their reward, i.e. their potential function, through a unilateral change of action. 

It is often feasible to choose appropriate forms of  agents' utility functions and also the global utility in order to enable the existence and use of  a potential function of the system. In this paper we assume that all players/robots share the same global reward. There are cases where, because of communication or other physical constraints, only reward functions that are shared among groups of robots can be used. But even in these cases   it is possible to create reward functions that can act as potentials. Wonderful life utility is such a utility function. It was introduced in \cite{wlu} and applied in \cite{autonomous} to formulate distributed optimisation tasks  as potential games. Player $i$'s utility, when wonderful life utility is used, can be defined as the difference between the global utility $r_{g}$ and the utility of the system when a reference action is used as Player $i$'s action. More formally when Player $i$ chooses an action $s^{i}$ one can define
\begin{equation}
r^{i}(s^{i})=r_{g}(s^{i},s^{-i})-r_{g}(s^{i}_{0},s^{-i}), \forall s^{i},\ \forall s^{-i}
\label{eq:wlu}
\end{equation}
where $s^{i}_{0}$ denotes a reference action for player $i$. A reference action is introduced because in strategic games there is no  action to represents the case when the player chooses to take no action.

\section{Rational agent cooperation }
\label{sect:rac}
\subsection{Agent Definitions}
\label{sec:background}


By analogy to previous definitions \cite{lincoln2013,wooldridge2009,rao_agents,rao_agents1} of AgentSpeak-like architectures, we define our agents as a tuple:
\begin{equation}
\label{eq:agent}
\mathcal{R}=\{ \mathcal{F},B,L,\Pi,A\}
\end{equation}
where:
\begin{itemize}
\item 
	$\mathcal{F} = \{p_1,p_2,\ldots,p_{n_p}\}$ is the set of all predicates.
\item
	$B \subset \mathcal{F}$ is the total set of belief predicates. The current belief base at time $t$ is defined as $B_t \subset B$. Beliefs that are added, deleted or modified can be either called \emph{internal} or \emph{external} depending on whether they are generated from an internal action, in which case are referred to as ``mental notes'', or  from an external input, in which case they are called ``percepts''. 
\item
	$L = \{l_1,l_2,\ldots\,l_{n_l}\}$ is a set of logic-based implication rules.
\item
	$\Pi = \{\pi_1,\pi_2,\ldots,\pi_{n_\pi}\}$ is the set of executable plans or \emph{plans library}. 

Current applicable plans at time $t$ are part of the subset applicable plan $\Pi_t \subset \Pi$ or "desire set".  
	
	 \item 
	
	$A = \{a_1,a_2,\ldots,a_{n_a}\} \subset \mathcal{F} \setminus B$ is a set of all available actions. Actions can be either \emph{internal}, when they modify the belief base or fata in memory objects, or \emph{external}, when they are linked to external functions that operate in the environment.
\end{itemize}


 AgentSpeak like languages, including the limited instruction set agent \cite{ecc16lisa},  can be fully defined and implemented by listing the following items:
\begin{itemize}
\item \emph{Initial Beliefs}.\\
	The initial beliefs and goals $B_0 \subset F$ are a set of literals that are automatically copied into the \emph{belief base} $B_t$ (that is the set of current beliefs) when the agent mind is first run.
\item \emph{Initial Actions}.\\
	The initial actions $A_0 \subset A$ are a set of actions that are executed when the agent mind is first run. The actions are generally goals that activate specific plans.

	\end{itemize}

The following three operations are repeated for each reasoning cycle. 
\begin{itemize}	
	\item \emph{Maintenance of Percepts}.  This means generation of perception predicates for $B_t$ and data objects such as the world model used here $W$.  
	\item \emph{Logic rules}.\\
	A set of logic based implication rules $L$ describes \emph{theoretical} reasoning to improve the agent current knowledge about the world.
\item \emph{Executable plans}.\\
	A set of \emph{executable plans} or \emph{plan library} $\Pi$. Each plan $\pi_j$ is described in the form:
	\begin{equation}
		p_j : c_j \leftarrow a_1, a_2, \ldots, a_{n_j}
	\end{equation}
where $p_j \in B$ is a \emph{triggering predicate}, which allows the plan to be retrieved from the plan library whenever it comes true, $c_j \in B$ is a logic formula of a \emph{context}, which allows the agent to check the state of the world, described by the current belief set $B_t$, before applying a particular plan sequence $a_1, a_2, \ldots, a_{n_j} \in A$ with a list of actions.  Each $a_j$ can be one of (1) predicate of an external action with arguments of names of data objects, (2) internal (mental note) with a preceding + or - sign to indicate whether the predicate needs to be added or taken away from the belief set $B_t$ (3) conditional set of items from (1)-(2).  The set of all triggers $p_j$ in a program is denoted by $E_{tr}$
\end{itemize}

The reasoning cycle of LISA used in this paper consists of the following steps (Fig. \ref{fig:lisacycle}):
\begin{enumerate}
\item
	\emph{Belief base update}.\\ The agent updates the belief base by retrieving information about the world through perception and communication. The job is done by two functions, called \emph{Belief Update Function (BUF)} and \emph{Belief Review Function (BRF)}. The BUF takes care of adding and removing beliefs from the belief base; the BRF updates the set of current events $E_t$ by looking at the changes in the belief base.
\item \emph{Application of logic rules}.\\	
	The logic rules in $L$ are applied in a round-robin fashion (restarting at the beginning of the list) until there are new predicates generated for $B_t$. This means that rules need to be verified not to lead to infinite loops.  
	\item
	\emph{Trigger Event Selection}.\\ For every reasoning cycle  a function called \emph{Intention Set Function} selects the current event set $E_t$:
	\begin{equation}
		S_t : \wp(B_t) \rightarrow \wp(E_t) \
	\end{equation}
	where $\wp(\cdot)$ is the so called \emph{power operator} and represents the set of all possible subset of a particular set. 
	We will call the current selected trigger event $S_t(B)=T_t$ and the associated plans the \emph{Intention Set}. 	
\item
	\emph{Plan Selection}.\\ All the trigger plans in $T_t$ are checked for their context  to form the  \emph{Applicable Plans} set $\Pi_t$,  	\begin{equation}
		S_O: E_t \rightarrow \wp(S_t)
	\end{equation}
	We will call the current selected plan $S_O(\Pi_t)=\pi_t$.	
\item
	\emph{Plan Executions}.\\ 
	 All plans in $S_O: E_t $ are started to be executed concurrently by going through the plan items 
	 $ a_1, a_2, \ldots, a_{n_j} $ one-by-one sequentially.   	 
\end{enumerate}

\begin{figure}[htbp]
	\centering
\includegraphics[width=80mm]{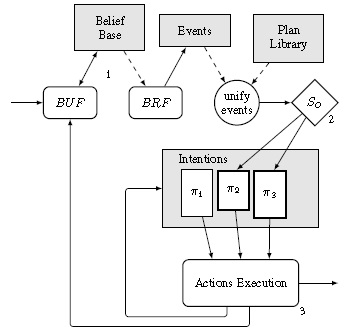}
\caption{Reasoning cycle of LISA: the plans run in a multi-threaded way, avoiding the need of $S_E$ and $S_I$}
	\label{fig:lisacycle}
\end{figure}
 
\subsection{Cooperation cycle}
\label{sect:coop_cycle}
The robots review their actions at a much slower rate than their reasoning cycle operates. During a cooperation cycle each of the agents observes the other's actions, learns from it, defines a reward function relevant to the environmental circumstances, performs optimisation to select its own  action takes or continues the  action undergoing.  In most practical situations the currently executed actions is one of the options for selection.   The agent theory described in the previous subsection can be used to make the process well coordinated among the agents using logic rules, perception cycles and  executable plans allocated as follows for each of the cyclical steps the agents need to perform collectively.  

Slower than the reasoning cycle, is the asynchronous \emph{operational cycle}.  Its length in time can vary depending on the agent and circumstances, it is asynchronous across the agent set as availability of communication is not guaranteed between all the agents at all times. Each agent executes the following steps during its operational cycle, which is an integer multiple $N_o\ge 3$ of its reasoning cycle.  Note that 
update of a simultaneous localisation and mapping-based world model $W$.  The following steps are carried out during the last of the reasoning cycles  in the operational cycle. This will allow enoug useful data to be
 
\begin{enumerate}
 \item  Agent $i$  performs estimation of other agent's actions set $s^{-i}$ by analysing changes in $W$ in time. The plan associated with this is  $$ estim\_agents : cooperate \leftarrow estim\_movements(W,Am), infer\_activities(Am,Sni), Opt(R,Sni,Act), trigger(Act). $$
  where $Am$ is available movement description data for the rest of the agents in the team and $Sni$ is a notation for $s^{-i}$.     Each item in $s\in s^{-i}$ can be a pair describing the situation and the action the agent is  taking.  The reward function $r^i(s^i,s^{-i})$ is represented by a data object $R$.  
  
\item  $trigger(Act)$ adds a predicate $pAct$ to $B_{t+1}$ to trigger the action plan associated with $Act$:
$$ pAct  :  \tilde  \ active(pAct)\  \&  \  \ \tilde  \    \ goal(reached)\  \leftarrow \ perform(pAct)  .   $$
Here the predicate $goal(reached)$ can  be established by logic inference in $L$ during reasoning cycles involving predicates generated by perception.

\item  Finally there is a mission completion plan to be automatically executed when the mission is completed. 
$$ goal(reached) : mission\_started  \leftarrow  return\_to\_base(B), -misson\_started.  $$
where $return\_to\_base(B)$ is a repeated action of returning with collision avoidance to a base described by data object $B$.
 
 \end{enumerate}   
 
  The above schema is generally applicable, can be particularised for its actions and  can be added to make and agent cooperative  by fictitious play.

\subsection{Example}

Consider the following case of a UAV team as an example. Assume that one of the UAVs in the team suddenly find itself with limited battery power and it can choose between two actions. One is that it quickly lands as this has small control cost and can therefore can  be performed safely. The second action is to pick up parcel  first, which has high control cost and  can result in unsafe landing later. In addition consider the case where the action of landing will affect negatively the performance of other UAVs and the team's mission will fail while picking up the parcel  would result in a successful performance of the desired joint task. The costs and rewards of this decision dilemma can be formulated in terms of the above mentioned energy and environmental costs.

In order to ensure that the robot team will accomplish its mission the impact of the control cost in robots' decisions should be smaller than the one of  the environmental cost. This can be achieved either by setting $w_{m}<<w_{e}$ or by setting the cost of a mission failure to be significantly higher than the maximum control cost. This example is to point out the variety of possible decisions that can be accommodated when reward functions are used to solve the distributed optimisation tasks. 

The following case study, which has also studied in simulation, shows how the proposed game theoretic decision making mechanism. We consider a team of robots who are to coordinate their team in order to identify possible threats in a warehouse, where we assume that there are $\mathcal{N}$ rooms with some hazardous items. The materials have different attributes of the following categories: flammable, chemical, radioactive and security sensitive. In each room  there can be items that belong up to two different categories but there are constraints.  For instance, are not allowed in the same room  flammable, radioactive and chemical materials together . Each of the $\mathcal{I}$ robots of the team is equipped with sensors which can sense the different attributes and capabilities. Figure \ref{fig:example} depicts an example of this scenario. The set of rooms correspond to the set of possible placement actions of each robot and is denoted by $C$. In each of these delivery states there are control actions of the robots  described by dynamical models. As an example we can consider the task of moving towards the region of interest. Each robot $i$ should choose one of the available rooms $n$, based on the cost of the total control costs of moving to room $n$, but also the estimated decisions of the other robots and the suitability of robot $i$ to examine room $n$ for its hazardous materials. 

\begin{figure}
\centering
\includegraphics[scale=0.25]{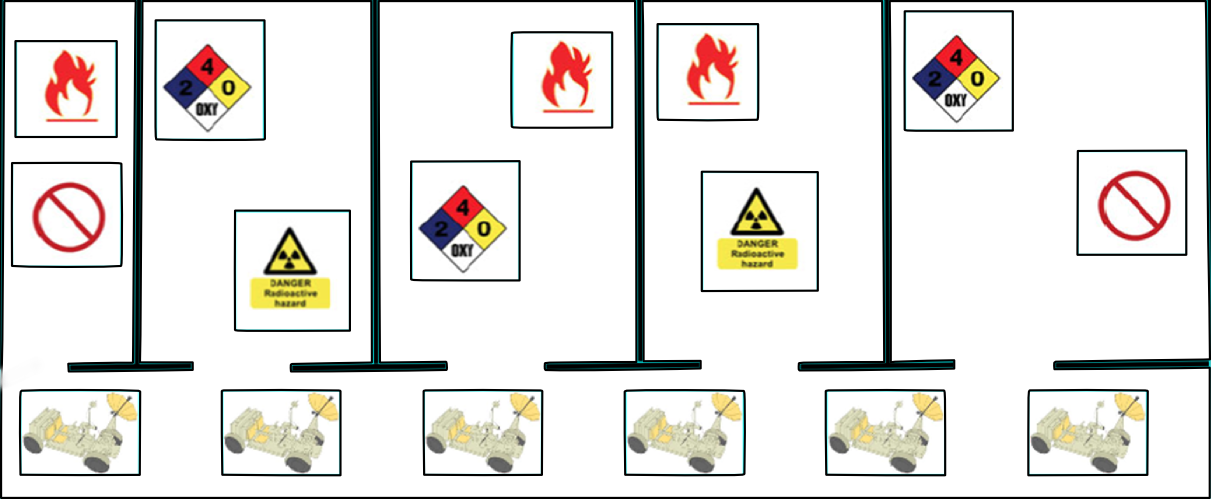}
\caption{Vehicles should patrol five areas of a warehouse. Each area has two objects of different attribute and therefore different value. The vehicles, according to their capabilities, should choose a joint action that will minimise their common cost function.}
\label{fig:example}
\end{figure}


In order to create the performance model of the proposed decision making process we need to define a utility function. Utility functions have been used as metrics of the robots coordination efficacy in applications, such as \citep{utcon1,utcon2,utcon3,utcon4,dta2}. Since the players of the game will maximise their expected reward the utility function can be seen as the negative of a cost function. Therefore we can include in the utility any control costs of the robots and environmental elements of the robots' tasks and the constraints that might arise from specific tasks. The function of the control cost should take into account the spatial characteristics of the problem like the cost to the robot to move towards to a specific position. The environmental part of the cost function takes into account costs that arose from the nature of the coordination problem and can also include the quality of the sensors of a robot, the aptness of the robots to perform specific tasks, etc.

We will use cooperative robot teams, robots who have different sensors and capabilities, but it is possible to use a similar controller in swarms of robots that have identical specifications. The differences between the robots can also be expressed in terms of their endurance in a specific environment, their efficiency to accomplish a specific task and the presence of the correct sensor to identify a specific threat. Fuzzy variables can be used to quantify robots' efficacy. Two fuzzy variables that can be used are sensors' quality and battery life. The values of battery life are short, fair, long and values of sensors' quality are low, medium and high. If a robot is not equipped with a specific sensor then its efficiency to detect an event is zero. Each robot, using a controller like the one that is depicted in Figure \ref{our_controller}, should then coordinate  with the other robots in order to efficiently choose a region to patrol.

The robots should coordinate and choose a region which they will patrol based on the possible threats that should be detected in each area, their sensors' specifications and the actions of the other robots. Each robot can choose only one region to patrol, but many robots can choose the same area. 


Each material in the warehouse has different significance and therefore each room has a different value depending on the objects that are stored there  which can be incorporated in $E_{i}(\cdot)$. In a game theoretic terminology we can define a potential game $\Gamma$ with $\mathcal{I}$ players which have $\mathcal{N}$ available actions. The utility function that is associated with room $n$ can be defined as:
\begin{displaymath}
r_{n}(s)=\sum_{\forall i, s^{i}=n }r_{i}^{m}(s^{i})+ \sum_{\forall s^{i}=n, n \in \mathcal{N} }r_{i}^{e}(s^{i},s^{-i}).
\label{eq:target_utility}
\end{displaymath}
where $r_{i}^{m}(s^{i})$, is a function that depends on the initial position of the robot and the room $n$, $n \in \mathcal{N} $ and represent the robot's cost to move toward a specific area and $r_{i}^{e}(s^{i},s^{-i})$ is the environmental part of the cost. 

Logic constraints such as the fact that a robot should not visit an area which it has not the appropriate sensors to patrol or that robots with the same attributes (sensors) should not choose the same area can be included in (\ref{eq:target_utility}) using penalty terms in $r_{i}^{e}(s^{i},s^{-i})$, i.e if a robot $i$ chooses to select an area which is not suitable of its sensors or then $r_{i}^{e}(s^{i},s^{-i})$ should have a negative or zero value. Similarly if two robots $i$ and $j$ with the same sensors choose the same area then $r_{i}^{e}(s^{i},s^{-i})$ and $r_{j}^{e}(s^{j},s^{-j})$ should also have a negative or zero value.      
The global utility that the robots will share is defined as:
\begin{equation}
r_{g}(s) =\sum_{n=1}^{n=\mathcal{N}}r_{n}(s)
\label{eq:gu}
\end{equation}

\section{The learning process}
\label{new_algorithm}
In this section we present a combination of fictitious play and extended Kalman filters as the algorithm that we will use in the learning block of the proposed decision making module. We briefly present the classic fictitious play algorithm and how it can be combined with extended Kalman filters in a decision making algorithm.

\subsection{Fictitious play}
Fictitious play \citep{brown_fict}, is a widely used learning technique in game theory. In fictitious play each player chooses
his action according to the best response of his beliefs about his opponent's
strategy.

Initially each player has some prior beliefs about the strategy
that his opponent uses to choose an action. These beliefs are expressed by a weighting function $\kappa_{t}^{i \rightarrow j}(s^j)$. Where $\kappa_{t}^{i \rightarrow j}(s^j)$ denotes the weight function of Player $i$ that Player $j$ will play action $s^{j}$ at the $t_{th}$ iteration.   The players, after each iteration, update the weight functions, and therefore their beliefs, about their opponents' strategy and
play again the best response to their beliefs. More formally, in the beginning of a game players maintain some arbitrary non-negative initial weight functions
\mbox{$\kappa_{0}^{i \rightarrow j}(s^j)$, $i=1, \ldots, \mathcal{I}$ , $j \in \{1,\ldots, \mathcal{I} \} \backslash \{i\}$} that are updated using the formula \citep{learning_in_games}:
\begin{equation}
\kappa_{t}^{i\rightarrow j}(s^{j}) = \kappa_{t-1}^{i\rightarrow j}(s^{j})+\mathfrak{I}_{s^j_t=s^j}
\label{eq:kappa}
\end{equation}
for each $j$, where
$\mathfrak{I}_{s^j_t=s^j}=\left\{\begin{array}{cl}1&\mbox{if$s^j_t=s^j$}\\0&\mbox{otherwise.} \end{array}\right\}$. 

Equation (\ref{eq:kappa}) suggests that all the observed actions have the same impact. Therefore players assume that their opponents choose their actions using a fixed mixed strategy. It is natural then to use a multinomial distribution to approximate an opponent's mixed strategy. The parameters of the multinomial distribution can be estimated using the maximum likelihood method. The mixed strategy of opponent $j$ is then estimated from the following formula:

\begin{equation}
\sigma_{t}^{i \rightarrow j}(s^{j})=\frac{\kappa^{i \rightarrow j}_{t}(s^j)}{\sum_{s' \in S^j}\kappa^{i \rightarrow j}_{t}(s')}.
\label{eq:fp1p}
\end{equation}


\subsection{Fictitious play as a state space model}
A more realistic assumption than using (\ref{eq:kappa}) is to presume that players are intelligent and change their strategies according to the other players' actions. We follow \cite{cj} and will represent the fictitious play process as a state-space model. According to the state space model each player has a propensity $Q_{t}^{i}(s^{i})$ to play each of their available actions $s^{i} \in S^{i}$, and then to  form a strategy based on these propensities. Finally players can choose  an action  based on their strategy and the best response decision rule. Because players have no information about the evolution of their opponents' propensities, and under the assumption that the changes in propensities are small from one iteration of the game to another, we model propensities using a Gaussian autoregressive prior on all propensities \citep{cj}. 
We set $Q_{0}\sim N(0,I)$, where $I$ is the identity matrix, and recursively update the value of  $Q_{t}$ according to the value of $Q_{t-1}$ as follows:
\begin{equation}
Q(s_{t})=Q(s_{t-1})+\eta_{t}
\label{eq:propen}
\end{equation}
where $\eta_{t}\sim N(0,\chi^{2}I)$. 

Similarly to the sigmoid function that is widely used in the neural network literature \cite{cbishop}, to relate the weights and the observation layer of a neural network, we assume that propensities are connected with players' actions by the following Boltzman formula for every $s^{j} \in S^{j}$
\begin{equation}
\mathfrak{I}_{s^j_t=s^j}=\frac{e^{(Q^{j}(s^{j})/\tau)}}{\sum_{\tilde{s} \in S^{j}}e^{(Q_{t}(\tilde{s})/\tau)}}.
\label{eq:meas}
\end{equation}

\subsection{Kalman filters and Extended Kalman filters}
Our objective is to estimate Player $i$'s opponent propensity and thus to estimate the marginal probability $p(Q_{t},s_{1:t})$. This objective can be represented as a Hidden Markov Model (HMM). HMMs are used to predict the value of an unobserved variable $c_{t}$, the hidden state,  using the observations of another variable $z_{1:t}$.  There are two main assumptions in the HMM representation. The former one is that the probability of being at any state $c_{t}$ at time $t$ depends only at the state of time $t-1$, $c_{t-1}$. The latter one is that an observation at time $t$ depends only on the current state $c_{t}$. One of the most common methods to estimate $p(c_{1:t},z_{1:t})$ is Kalman filters and its variations. Kalman filter \citep{kalman} is based on two assumptions, the first is that the state variable is Gaussian. The second is that the observations are the result of a linear combination of the state variable. Hence Kalman filters can be used in cases which are represented as the following state space 
model:
\begin{equation}
\begin{array}{rl}
c_t= &  Ac_{t-1}+\xi_{t-1}  \textrm{   hidden layer} \\
y_{t}= & Bc_{t}+\zeta_{t} \textrm{   	observations} 
\end{array}
\end{equation}
where $\xi_t$ and $\zeta_t$ follow a zero mean normal distribution with covariance matrices $\Xi=q_{t}I$ and $Z=r_{t}I$ respectively, and $A$, $B$ are linear transformation matrices. When the distribution of the state variable $c_t$ is Gaussian then $p(c_{t}|y_{1:t})$ is also a Gaussian distribution, since $y_{t}$ is a linear combination of $c_{t}$. Therefore it is enough to estimate its mean and variance to fully characterise $p(c_{t}|y_{1:t})$. 

Nevertheless in the state space model we want to implement, the relation between Player $i$'s opponent propensity and his actions is not linear (\ref{eq:meas}). Thus a more general form of state space model should be used, such as:  
\begin{equation}
\begin{array}{rl}
c_t&=f(c_{t-1})+\xi_{t} \\
y_{t}&=h(x_{t})+ \zeta_{t} 
\end{array}
\label{eq:ekf_state}
\end{equation}
where $f(\cdot)$ and $h(\cdot)$, are non-linear functions, $\xi_{t}$ and $\zeta_{t}$ are the hidden and observation state noise respectively, with zero mean and covariance matrices $\Xi=q_{t}I$ and $Z=r_{t}I$ respectively. The distribution of $p(c_{t}|y_{1:t})$ is not a Gaussian distribution because $f(\cdot)$ and  $h(\cdot)$ are non-linear functions. Extended Kalman filter (EKF) provides a simple method to overcome this shortcoming by using a first order Taylor expansion to approximate the distributions of the sate space model in (\ref{eq:ekf_state}). In particular if we let $c_{t}=m_{t-1}+\epsilon$, where $m_{t}$ denotes the mean of $c_{t}$ and $\epsilon \sim N(0,P)$, we can rewrite (\ref{eq:ekf_state}) as: 
\begin{equation}
\begin{array}{rl}
c_t&=f(m_{t-1}+\epsilon)+w_{t-1}=f(m_{t-1})+F_{c}(m_{t-1})\epsilon +\xi_{t-1}\\
y_{t}&=h(m_{t}+\epsilon)+\zeta_{t}=h(m_{t})+H_{c}(m_{t})\epsilon+\zeta_{t}
\end{array}
\label{eq:ekf_taylor}
\end{equation}
where $F_{c}(m_{t-1})$ and $H_{c}(m_{t})$ is the Jacobian matrix of $f$ and $h$ evaluated at $m_{t-1}$ and $m_{t}$, respectively. If we use the transformations of (\ref{eq:ekf_taylor}) then $p(c_{t}|y_{1:t})$ is a Gaussian distribution.

Since $p(c_{t}|y_{1:t})$ is a Gaussian distribution to fully characterise it we need to evaluate its mean and its variance. The EKF process \citep{ekf1,ekf2} estimates this mean and variance in two steps the prediction and the update step. In the prediction step at any iteration $t$ the distribution of the state variable is estimated based on all the observations until time $t-1$, $p(c_{t}|y_{1:t-1})$. The distribution of $p(c_{t}|y_{1:t-1})$ is Gaussian and we will denote its mean and variance  as  $m_{t}^{-}$ and $P_{t}^{-}$ respectively. During the update step the estimation of the prediction step is corrected in the light of the new observation at time $t$, so we estimate $p(c_{t}|y_{1:t})$. This is also a Gaussian distribution and we will denote its mean and variance  as  $m_{t}$ and $P_{t}$ respectively.

The estimate of the state $c$ for and its variance prediction for $p(c_{t}|y_{1:t-1})$ and $p(c_{t}|y_{1:t})$ are evaluated based on the prediction and update steps of the EKF process \citep{ekf1,ekf2} respectively as follows: \\
\textbf{Prediction Step}
\begin{align}
\label{eq:pdstep}
\bar{c}_{t}= &f(\hat{c}_{t-1}) \nonumber \\
P_{t}^{-}=&F(m_{t-1})P_{t-1}F(m_{t-1})+W_{t-1} \nonumber
\end{align}
where $\bar{c}_t$ and $\hat{c}_{t-1}$ are the estimates of $c$ for $p(c_{t}|y_{1:t-1})$ and $p(c_{t-}|y_{1:t-1})$ respectively, and the $j,j'$ element of $F(m_{t})$ is defined as
\begin{displaymath}
[F(m_{t}^{-})]_{j,j'}=\frac{\partial f(c_{j},r)}{\partial c_{j'}}\arrowvert_{c=m_{t}^{-}, q=0}	
\end{displaymath}
\textbf{Update Step}
\begin{eqnarray}
\label{eq:updatestep}
v_t&=&z_{t}-h(\bar{c}_{t}) \nonumber \\
S_{t}&=&H(\bar{c}_{t})P_{t}^{-}H^{T}(\bar{c}_{t})+Z \nonumber \\
K_{t}&=&P_{t}^{-}H^{T}(\bar{c}_{t})S_{t}^{-1} \nonumber \\
\hat{c}_{t}&=&\bar{c}_{t}+K_{t}v_{t}\nonumber \\
P_{t}&=&P_{t}^{-}-K_{t}S_{t}K_{t}^{T} \nonumber
\end{eqnarray}
where $z_{t}$ is the observation vector and the $(j,j)'$ element of $H(c_{t})$ is defined as:  
\begin{displaymath}
[H(m_{t}^{-})]_{j,j'}=\frac{\partial h(c_{j},r)}{\partial c_{j'}}\arrowvert_{c=m_{t}^{-}, r=0}	
\end{displaymath}

\subsection{Fictitious play and EKF}

For the rest of this paper we will only consider inference over a single opponent mixed
strategy in fictitious play. Separate estimates will be formed
identically and independently for each opponent. We therefore
consider only one opponent, and we will drop all dependence on player $i$, and
write $s_{t}$, $\sigma_{t}$ and $Q_{t}$ for Player $i$'s opponent's action, strategy and propensity respectively. Moreover for any vector $x$, $x[j]$ will denote the $j_{th}$ element of the vector and for any matrix $y$, $y[i,j]$ will denote the $(i,j)_{th}$ element of the matrix.

We can use the following state space model to describe the fictitious play process:
\begin{equation}
\begin{array}{cc}
Q_{t}&=Q_{t-1}+\xi_{t-1} \\
\mathfrak{I}_{s^j_t=s^j}&= h(Q_{t})+\zeta_{t} 
\end{array}
\label{eq:fpekf}
\end{equation}
where $\xi_{t-1} \sim N(0,\Xi)$, is the noise of the state process and $\zeta_{t} \sim N(0,Z)$ is is the error of the observation state with zero mean and covariance matrix $Z$, which occurs because we approximate a discrete process like best response (\ref{eq:br}), using a continuous function $h(\cdot)$, where $h_{s(k)}(Q)=\frac{exp(Q_{t}[s(k)]/\tau)}{\sum_{\tilde{s} \in S} exp(Q_{t}[\tilde{s}]/\tau)}$, where $\tau$ is a temperature parameter. Hence we can combine the EKF with fictitious play as follows.

\begin{table*}
\begin{center}
\begin{tabular}{p{12cm}}
\hline
\hline
\begin{enumerate}
\item At time $t$ agent $i$ maintains   estimates  of his opponent's propensities up to time $t-1$, $\hat{Q}^j_{t-1}$, with covariance $P^j_{t-1}$ of his distribution.

\item Agent $i$ predicts his opponents' propensities $\bar{Q}_{tk}^j,\ j\in \{1,...,\mathcal{I}\}\backslash\{i\},\ k\in S^j$  using (\ref{eq:predekffp}).

\item Based on the propensities in 2 each agent updates his beliefs about his opponents' strategies using (\ref{eq:strategies}).

\item Agent $i$ chooses an action based on the beliefs in 3 and applies best response decision rule.

\item The agent $i$  observes his opponents' action $s^j_{tk},\ j\in \{1,...,\mathcal{I}\}\backslash\{i\}$.

\item The agent update its estimates of all of its opponents' propensities using extended Kalman Filtering to obtain $\hat{Q}^j_t,\ j\in \{1,...,\mathcal{I}\}\backslash\{i\}$.
\end{enumerate}
\\
\hline
\hline
\end{tabular}
\caption{EKF based fictitious play algorithm.}
\label{skata}
\end{center}
\end{table*}

At time $t-1$ Player $i$ has observed action $s_{t-1}$ and based on the update step of the EKF process has an estimate of his opponent's propensity, $\hat{Q}_{t-1}$ with the variance $P_{t-1}$. Then at time $t$ he uses EKF prediction step to estimate his opponent's propensity. The estimate and its variance are:
\begin{equation}
\begin{array}{r l}
\bar{Q}_{t}=&\hat{Q}_{t-1} \\
P_{t}^{-}=&P_{t-1}+\Xi 
\end{array}
\label{eq:predekffp}
\end{equation}
Player $i$ then evaluates his opponents strategies using his estimations as:
\begin{equation}
 \sigma_{t}(s_{t}=k)=\frac{exp(\bar{Q}(s_{t}=k) / \tau)}{\sum_{\tilde{s} \in S}exp(\bar{Q}(\tilde{s}) /\tau)}.
\label{eq:strategies}
\end{equation}
Player $i$ then uses the estimate of his opponent strategy (\ref{eq:strategies}) and best responses (\ref{eq:br}), to choose an action. After observing his opponent's action $s_{t}$, Player $i$ correct his estimate about his opponent's propensity using the update equations of EKF process. The update equations are:
\begin{eqnarray}
v_t&=&\mathfrak{I}_{s^j_t=s^j}-h(\bar{Q}_{t}) \nonumber \\
S_{t}&=&H(\bar{Q}_{t})P_{t}^{-}H^{T}(\bar{Q}_{t})+Z \nonumber \\
K_{t}&=&P_{t}^{-}H^{T}(\bar{Q}_{t})S_{t}^{-1} \nonumber \\
\hat{Q}_{t}&=&\bar{Q}_{t}+K_{t}v_{t} \nonumber \\
P_{t}&=&P_{t}^{-}-K_{t}S_{t}K_{t}^{T} \nonumber
\label{eq:fpupdate}
\end{eqnarray}
The Jacobian matrix $ H(\bar{Q}_{t})$ is defined as \\ $[H(\bar{Q}_{t})]_{j,j'}=\left\{\begin{array}{cl} \frac{\sum_{j \neq j'}\exp(\bar{Q}_{t}[j])\exp(\bar{Q}_{t})}{(\sum_{j}\exp(\bar{Q}_{t}[j]))^2}&\mbox{if
$j=j'$}\\- \frac{\exp(\bar{Q}_{t}[j])\exp(\bar{Q}_{t}[j'])}{(\sum_{j}\exp(\bar{Q}_{t}[j]))^2} &\mbox{if $j \neq$ j'}\end{array}\right.$. 

Table \ref{skata} summarises the fictitious play algorithm when EKF is used to predict opponents strategies. 

\subsection{EKF fictitious play as part of the reasoning cycle of Jason}
As described in Section \ref{sect:coop_cycle} robots actions and reasoning cycles are not synchronised. After each action of the robot team the utility function can be updated and a new reasoning cycle that will include the EKF fictitious play learning algorithm can be started. In particular steps 1 and 2 of the the reasoning cycle of Jason, (Section \ref{sec:background}), can be updated using the proposed learning process. The update of the propensities can be seen as \emph{Belief base update}, since when agents update their propensities, they update their beliefs about other agents actions for a given state of the environment. Then the action selection of EKF fictitious play is the \emph{Trigger Event Selection}. Note here that instead of Best Response other alternatives can be used with EKF fictitious play algorithm as part of the \emph{Trigger Event Selection} such as Smooth Best Response \cite{learning_in_games}.   

\section{The Main Results}
\label{theory}
In this section we present our convergence results for games with at least one pure Nash equilibrium. The results  are for the EKF fictitious play algorithm of Table \ref{skata}, when the covariance matrices $\Xi$ and $Z$ are defined as $\Xi=(\tilde{\xi}+\epsilon)I$ and $Z=(1/t)I$ respectively, where $\tilde{\xi}$  is a constant, $\epsilon$ is an arbitrarily small Gaussian random variable, $\epsilon \sim N(0,\Psi)$, $t$ is the $t_{th}$ iteration of fictitious play, and $I$ is the identity matrix. 

The EKF fictitious play algorithm has the following properties: 
\begin{prop}
\label{prop1}
If at iteration $t$ of the EKF fictitious play algorithm, action $k$ is played from Player $i$'s opponent, then the estimate of his opponent propensity to play action $k$ increases, $\hat{Q}_{t-1}[k]<\hat{Q}_{t}[k]$. Moreover if we denote $\Delta[i]$ as  $\Delta[i]=\hat{Q}_{t}[i]-\hat{Q}_{t-1}[i]$, then $\Delta[k]>\Delta[j] \forall j \in S^{j}$, where $S^{j}$ is the action space of the $j_{th}$ opponent of Player $i$. Therefore since $\sum_{k\in S^{j}}{\sigma(s^{j}=k)}=1$, $\sigma_{t}^{j}$ will be also increased. 
\end{prop}

\begin{proof}
The proof of Proposition \ref{prop1} is on Appendix \ref{append1}.
\end{proof}
Proposition \ref{prop1} implies that player $i$, when he uses EKF fictitious play, learn his opponent's strategy and eventually they will choose the action that will maximise their reward base on their estimation. Nevertheless there are cases where players may change their action simultaneously and become trapped in a cycle instead of converging in a pure Nash equilibrium. As an example we consider the symmetric game that is depicted in Table \ref{tab:simcoord}.      
This is a simple coordination game with two pure Nash equilibria the joint actions $(U,L)$ and $(D,R)$. In the case were the two players start from joint action $(U,R)$ or $(D,L)$ and they always change their action simultaneously then they will never reach one of the two pure Nash equilibria of the game.

\begin{prop}
\label{prop2}
When the players of a game $\Gamma$ use EKF fictitious play process to choose their actions, then with high probability they will not change their action simultaneously infinitely often. 
\end{prop}

\begin{proof}
The proof of Proposition \ref{prop2} is on Appendix \ref{append2}.
\end{proof}

Based on Proposition \ref{prop1} and \ref{prop2} we can infer the following propositions and theorem. 

\begin{prop}
\label{prop3}
(a) If $s$ is a pure Nash equilibrium of a game $\Gamma$, and $s$ is played at date $t$ in the process of EKF fictitious play, $s$ will be played at all subsequent dates. That is, strict Nash equilibria are absorbing for the process of EKF fictitious play.

(b) Any pure strategy steady state of EKF fictitious play must be a Nash equilibrium.
\end{prop}

\begin{proof}
Consider the case where players beliefs $\hat{\sigma}_{t}$, are such that their optimal choices correspond to a strict Nash equilibrium $\hat{s}$. In EKF fictitious play process players' beliefs are formed identically and independently for each opponent based on equation (\ref{eq:strategies}). By Proposition \ref{prop1} we know that players' estimations about their opponents' propensities and therefore their strategies will increase for the actions that are included in $\hat{s}$. Thus the best response to their beliefs $\hat{\sigma}_{t+1}$ will be again $\hat{s}$ and since $\hat{s}$ is a Nash equilibrium they will not deviate from it. Conversely, if a player remains at a pure strategy profile, then eventually the assessments will become concentrated at that profile, because of Proposition \ref{prop1} and so if the profile is not a Nash equilibrium, one of the players would eventually want to deviate.
\end{proof}

\begin{prop}
\label{prop4}
Under EKF fictitious play, if the beliefs over each player's choices converge, the strategy profile corresponding to the product of these distributions is a Nash equilibrium. 
\end{prop}

\begin{proof}
Suppose that the beliefs of the players at time t, $\sigma_{t}$, converges to some
profile $\hat{\sigma}$. If $\hat{\sigma}$ were not a Nash equilibrium, some player would eventually want
to deviate and the beliefs would also deviate since based on Proposition \ref{prop1} players eventually learn their opponents actions.
\end{proof}

Based on the propositions (\ref{prop1}-\ref{prop4}) we can show that EKF fictitious play converges to the the Nash equilibrium of games with a better reply path. A game with a better reply path  can be represented as a graph were its edges are the join actions of the game $s$ and there is a vertex that connects $s$ with $s'$ iff only one player $i$ can increasing his payoff by changing his action \citep{payton_young}. Potential games have a better reply path \citep{payton_young}. 
\begin{thm}
\label{theo2}
The EKF fictitious play process converges to the Nash equilibrium in games with a better reply path.  
\end{thm}

\begin{proof}
 If the initial beliefs of the players are such that their initial joint action $s_{0}$ is a Nash equilibrium, from Proposition \ref{prop3} and equation (\ref{eq:strategies}), we know that they will play the joint action which is a Nash equilibrium for the rest of the game. 
  
Moreover,   the initial beliefs of the players are such that their initial joint action $s_{0}$ is not a Nash equilibrium based on Proposition \ref{prop1} and Proposition \ref{prop2} after a finite number of iterations because the game has a better reply path the only player that can improve his pay-off by changing his actions will choose a new action which will result in  a new joint action $s$. If this action is not the a Nash equilibrium then again after finite number of iterations the player who can improve his pay-off will change action and a new joint action $s'$ will be played. Thus after the search of the vertices of a finite graph, and therefore after a finite number of iterations, players will choose a joint action which is a Nash equilibrium. Eventually after a finite number of time steps, $T$, the process will end up in a pure Nash equilibrium. The maximum number of iterations that is needed is the cardinality of the joint action set multiplied with the total number of iterations that is needed in order not to have simultaneous changes, $\binom{n}{2}(t_1+t_2+t_3+\ldots+T)$
\end{proof}

\section {Parameters selection}
\label{parameters}
The covariance matrix of the state space error $\Xi=\tilde{\xi}I$ and the measurement error $Z=\tilde{\zeta}I$ are two parameters that we should define in the beginning of the EKF fictitious play algorithm and they affect its performance.  Our aim is to find values, or range of values, of $\tilde{\xi}$ and $\tilde{\zeta}$ that can efficiently track opponents' strategy when it smoothly or abruptly change, instead of choosing them heuristically for each opponent when we use the EKF algorithm. Nevertheless it is possible that for some games the results of the EKF algorithm will be improved for other combinations of $\tilde{\xi}$ and $\tilde{\zeta}$ than the ones that we propose in this section. 

We examine the impact of EKF fictitious play algorithm parameters in its performance, by measuring the mean square error (MSE) between EKf estimation and the real opponent strategy, in the following two tracking scenarios. In the former scenario a single opponent chooses his actions using a mixed strategy which changes smoothly and has a sinusoidal form over the iterations of the tracking scenario. In particular for $t=1, 2, \ldots, 5000$ iterations of the game: $\sigma_{t}(1)=\frac{cos\frac{{2\pi t}}{n}+1}{2}=1-\sigma_t(2)$, where $n=100$. In the second toy example Player $i$'s opponent change his strategy abruptly and chooses action 1 with probability $\sigma_{t}^{2}(1)=1$ during the first and the last 1250 iterations' of the game and for the rest iterations of the game $\sigma_{t}^{2}(1)=0$. The probability of the second action is calculated as: $\sigma_{t}^{2}(2)=1 - \sigma_{t}^{2}(1)$. 

Figure \ref{fig:contour1} depicts the  mean square error for both tracking scenarios for $0.005\leq\tilde{\xi}\leq0.5$ and $0.005\leq\tilde{\zeta}\leq0.5$. Additionally we examine the performance of EKF fictitious play for $\tilde{\zeta}=1/t$  which is depicted as the last element of the x-axis of Figure \ref{fig:contour1}. There is a wide range of combinations of $\tilde{\xi}$ and $\tilde{\zeta}$ that minimise the MSE where the mean for both examples. The MSE is minimised for a wider range of $\tilde{\xi}$ when $\tilde{\zeta}=1/t$ than the all the cases where a single value of $\tilde{\zeta}$  was used. Moreover the minimum value of the MSE was observed when $\tilde{\zeta}=\frac{1}{t}$ and $\tilde{\xi}=0.1$. Figures \ref{fig:toy1} and \ref{fig:toy2} depict the opponent's strategy tracking of EKF fictitious play in both examples for these parameters. 
\begin{figure}%
\begin{center}
	\includegraphics[scale=0.6]{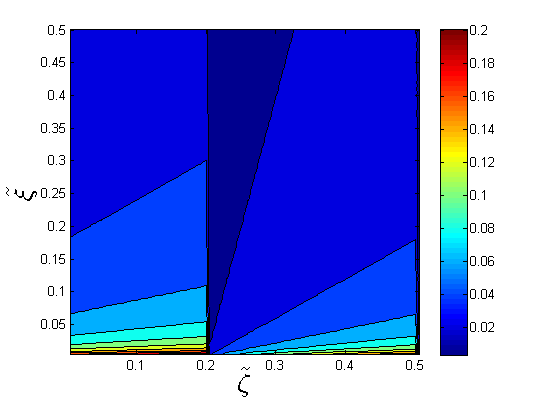}%
\caption{Contour plot of the combined mean square error for the pure and mixed strategy case. The range of $\xi$ and
$\zeta$ is \mbox{[0.005, 0.5] }. The last value of x-axis corresponds to $\tilde{\zeta}=1/t$ the }%
\label{fig:contour1}%
\end{center}
\end{figure}
\begin{figure}%
\begin{center}
	\includegraphics[scale=0.7]{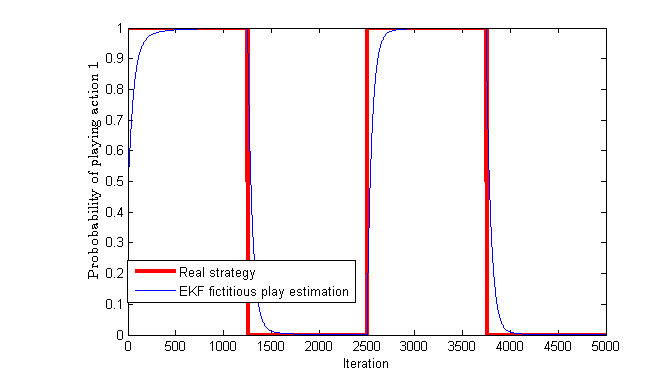}%
\caption{Tracking of opponent strategy at a pure strategy environment, where the pre-specified strategy of the opponent and its prediction are depicted as the red and blue line respectively.}%
\label{fig:toy1}%
\end{center}
\end{figure}
\begin{figure}%
\begin{center}
	\includegraphics[scale=0.7]{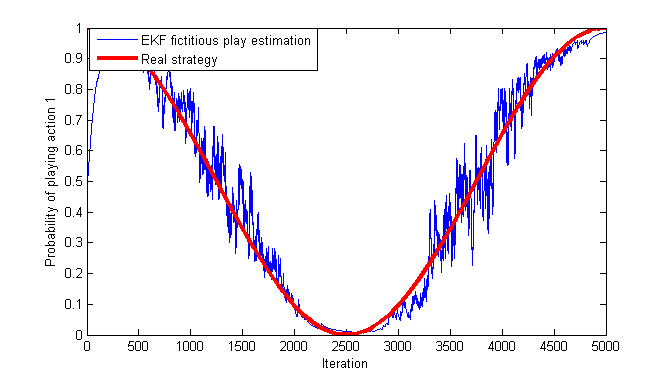}%
\caption{Tracking of opponent strategy at a mixed strategy environment, where the pre-specified strategy of the opponent and its prediction are depicted as the red and blue line respectively.}%
\label{fig:toy2}%
\end{center}
\end{figure}

\section{Simulations}
\label{simulation}
This section contains the results of two simulation scenarios, on which the performance of the proposed learning algorithm was tested. Based on the results of the previous section, we used the following values for $\Xi$ and $Z$ in our simulations,  $\Xi_t=(\tilde{\xi}+\epsilon)I$ and $Z_{t}=\tilde{\zeta}I$, where $I$ is the identity matrix, $\tilde{\xi}=0.1$, $\tilde{\zeta}=1/t$ and $\epsilon$ is an arbitrarily small Gaussian random number. For both simulation scenarios we report the average results for 100 replications of each learning instance. In each learning instance the agents were negotiating for 50 iterations.

\subsection{Results in symmetric games}

We initially tested the performance of the proposed algorithm in symmetric games  as they are presented in Tables \ref{tab:simcoord} and \ref{tab:3psym}. Even though these are games that involve only two players, they can be ensued from realistic applications. 
\begin{table}
\centering
\begin{tabular}{|c| c| c|c|}
\hline
 &Weak&Fair&Strong\\ \hline
 Weak& 1,1 & 0,0 & 0,0  \\ \hline
 Fair& 0,0 & 1,1 & 0,0 \\ \hline
 Strong& 0,0& 0,0 &1,1 \\ \hline 
 \end{tabular}
\caption{Players' rewards of a symmetric game.}
 \label{tab:3psym}
\end{table}
The game of Table \ref{tab:simcoord} can be seen as a collision avoidance game. Consider the case where  two UAVs, the players of the game, should fly in opposite directions. They will gain some reward if they accomplish their task and fly in different altitudes. Therefore players should coordinate and choose  one of the two joint actions that maximise their reward, $UL$ and $DR$, who represent the choices of different altitudes from the UAVs.  

Another example of how a symmetric game can be ensued from a realistic application comes from the area of material handling robots. Consider the case where two robots should coordinate in order to move some objects to a desired destination. The robots can either push or pull the objects depending on the direction of the destination of the object. Moreover, each robot can apply different forces to the objects. The amount of the force that a robot applies to an object can be represented as a fuzzy variable that can take the values: Weak, Fair and Strong. This scenario can lead to the game that is described in Table \ref{tab:3psym}.

Since in this paper we focus on the decision making module of a robot's controller, we will compare the proposed learning algorithm of EKF fictitious play with the particle filter alternative \cite{cj}. We will use  it as a baseline criterion for the greedy choice of the players, i.e the action that maximises the reward of a player independently of what the others are doing. In particular for the games that are presented in Tables \ref{tab:simcoord} and \ref{tab:3psym} this is the random action since the rewards for all the actions are the same.

Figures \ref{fig:res1} and \ref{fig:res2} depict the percentage of the replications where the algorithms converged to one of the optimal solutions for the games depicted   in Tables \ref{tab:simcoord} and \ref{tab:3psym} respectively. Because of the symmetry in the utility function, the particle filter fictitious play behaves similarly to the the random action. Thus when the initial beliefs were not supporting one of the pure Nash equilibria of the games, it failed to converge to joint action with reward. EKF fictitious play on the other hand, always converged to a joint action that maximises players rewards.
\begin{figure}%
\begin{center}
	\includegraphics[scale=0.65]{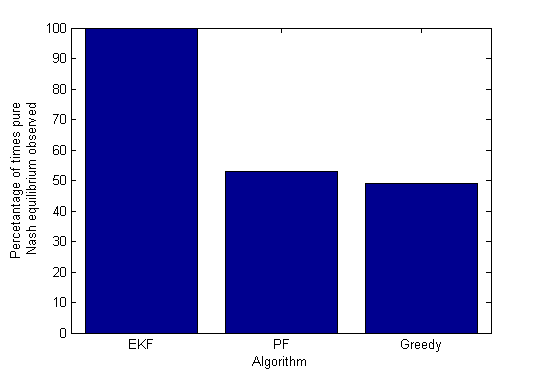}%
\caption{Percentage of the times that each algorithm converged to one of the two optimum joint actions for the game of Table \ref{tab:simcoord}.}%
\label{fig:res1}%
\end{center}
\end{figure}
\begin{figure}%
\begin{center}
	\includegraphics[scale=0.65]{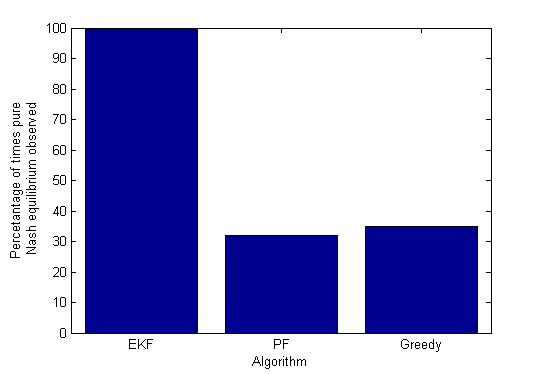}%
\caption{Percentage of the times that each algorithm converged to one of the three optimum joint actions for the game of Table\ref{tab:3psym}.}%
\label{fig:res2}%
\end{center}
\end{figure}
Figure \ref{fig:res3} shows the average utility that players obtain in each iteration of the game when the initial joint actions have zero reward. As we can observe, after the maximum of 35 iterations, the utility that players receive is 1 and therefore it converges to one of the pure Nash equilibria.
\begin{figure}%
\begin{center}
	\includegraphics[scale=0.6]{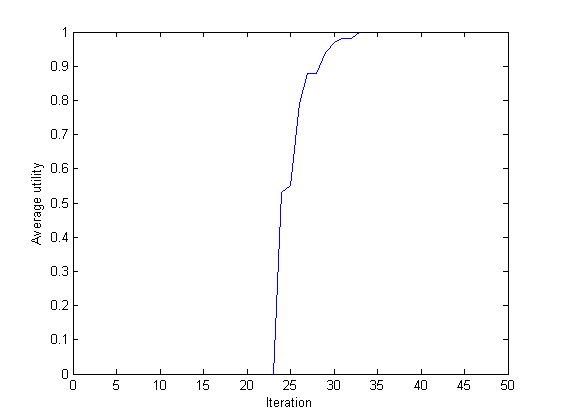}%
\caption{Average reward that players obtain when their initial action is either $UR$ or$DL$ for the game of Table \ref{tab:simcoord} .}%
\label{fig:res3}%
\end{center}
\end{figure}

\subsection{Results in a robots team  warehouse patrolling task}

We also examined the performance of EKF fictitious play in the task allocation scenario we described in Section \ref{sec:games}. A robot team, which consists of $\mathcal{I}$ robots, should examine an area $A$ that is divided in $\mathcal{N}$ regions with hazardous objects. We present simulation results for various combinations of number of regions and sizes of the robot's team. In particular, we examined the performance of EKF learning algorithm when the robot's team consists of $\mathcal{I}$ robots, $\mathcal{I}\in \{5,10,15,20,25,30,35,40\}$, who should patrol one of the $\mathcal{N}$ regions, $\mathcal{N}\in\{5,10,20\}$. 

Each region contains two objects with different attributes. Each robot $i$ then can be equipped with up to three of the following sensors: fire detector, chemical detector, Geiger counter and vision system. Therefore robots with a fire detector should patrol areas with flammable objects, robots with Geiger should patrol areas with radioactive material etc. Moreover each robot can move towards a region $n$ with a velocity $\mathcal{V}_{in}$. In our case study we assume that the velocity of a robot $i$, towards a region $n$ can be either slow, medium or fast. Therefore the time that a robot needs to reach a region $n$,  $T_{in}$, depends both in the distance between the robot and the region and the velocity that the robot will choose to move towards $n$ as well.  The global reward that is shared among the robots is defined as:

\begin{equation}
\begin{array}{r l}
r_g(s)=&- \sum_{n=1}^{\mathcal{N}}\sum_{i:s^{i}=n}c_{1}T_{in}+c_{2}\mathcal{V}_{in}\\
& +\sum_{n=1}^{\mathcal{N}}\sum_{k\in K_{n}} V_{nk}(1-\prod_{i:s^{i}=n}(1-\mathcal{E}_{ink}))
\end{array}
\label{eq:disut}
\end{equation}
where $c_{1}=1$, $c_{2}=1/6$ and $\mathcal{E}_{ink}$ is a metric of robot $i$'s efficiency to patrol a region $n$ based on its battery capacity and its sensor quality. It can also be seen as the probability that robot $i$ has to detect a threat $k$ in region $n$. We define $\mathcal{E}_{ink}$ as a function of robot $i$'s battery capacity and its sensors' quality as it is presented in Table \ref{tab:3prob}. If a robot has no suitable sensor to patrol a specific area $\mathcal{E}_{ink}=0$. 
\begin{table}
\centering
\begin{tabular}{|c| c| c c c|}
\hline 
\hline
 & & \multicolumn{3}{c|}{Battery capacity}\\ \cline{3-5}
 & & Short &Fair&Long \\ \cline{2-5}
 \multirow{3}{*}{Sensors' quality}&Low & 0.2 & 0.3&0.5  \\ \cline{3-5}
 &Medium& 0.3 & 0.5 & 07 \\ \cline{3-5}
 &High& 0.5& 0.7 &0.9 \\ \hline
 \end{tabular}
\caption{Efficiency that a robot efficiently patrol the area for a specific threat.}
 \label{tab:3prob}
\end{table}
Note that the first term in he summation represents the control cost the robots have when they choose a specific region to allocate and the second term corresponds to the environmental cost.

Even though this article considers distributed optimisation tasks, a centralised alternative, genetic algorithm (GA), will be also used. Since the centralised optimisation algorithms are expected to perform equally good or better than distributed optimisation algorithms, the performance of GA can be seen as a benchmark solution. Thus a central decision maker will choose in advance the area that each robot will patrol and the robot completes the pre-optimised allocation task. In our comparison we used the genetic algorithm function of Matlab's optimisation toolbox. Since genetic algorithms are stochastic processes, in each repetition of the same replication of the task allocation problem it will converge to a different solution. Thus we used two instances of the genetic algorithm, in the former one we used a single repetition of the genetic algorithm per replication of our problem, in the latter one for each replication of the task allocation problem we chose the allocation of the maximum utility among 50 runs of the genetic algorithm. 

In order to be able to average across the 100 replications, we used the following scores for the global rewards players obtained for each learning algorithm they used: $\mathcal{F}=100 \times \frac{r_{g}(s)}{r_{max}}$, where $r_{max}$ is the reward that players would have obtained if all all areas had been patrolled with 100\% efficiency.

Figures \ref{fig:res4}, \ref{fig:res5} and \ref{fig:res6} depict the score players obtain in the final iteration of the algorithm as a function of the number of robots, $\mathcal{I}$, for the tested algorithms when there are 5,10 and 20 areas of interest respectively.

\begin{figure}%
\begin{center}
	\includegraphics[scale=0.8]{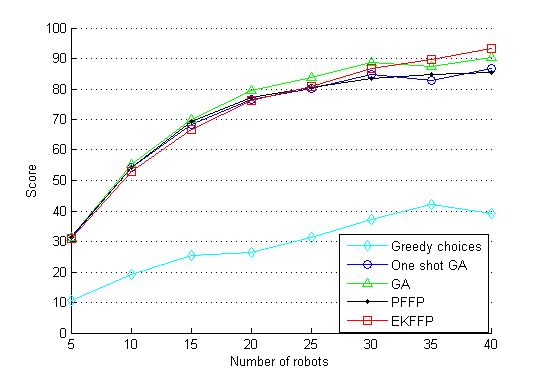}%
\caption{Robot team task allocation outcome for the case where $\mathcal{N}=5$. The x-axis represent the number of robots that was used in the simulation scenario and the y-axis the corresponding average score of the 100 replications of the allocation task. }%
\label{fig:res4}%
\end{center}
\end{figure}

\begin{figure}%
\begin{center}
	\includegraphics[scale=0.8]{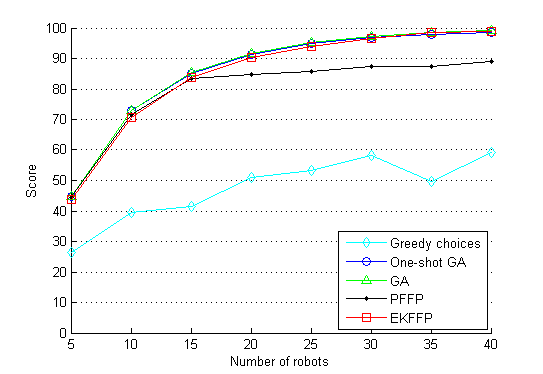}%
\caption{Robot team task allocation outcome for the case where $\mathcal{N}=10$. The x-axis represent the number of robots that was used in the simulation scenario and the y-axis the corresponding average score of the 100 replications of the allocation task.}%
\label{fig:res5}%
\end{center}
\end{figure}

\begin{figure}%
\begin{center}
	\includegraphics[scale=0.8]{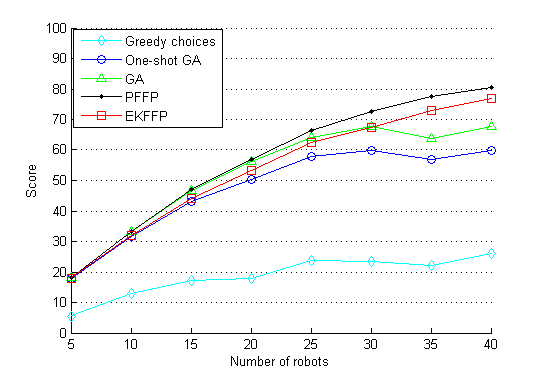}%
\caption{Robot team task allocation outcome for the case where $\mathcal{N}=20$. The x-axis represent the number of robots that was used in the simulation scenario and the y-axis the corresponding average score of the 100 replications of the allocation task.}%
\label{fig:res6}%
\end{center}
\end{figure}

In all cases when players chose the greedy actions, the outcome was very poor, which also indicates the need of the coordination controller that we propose when distributed solutions are considered. When the case with 5 areas is considered, Figure \ref{fig:res4}, all algorithms performed similarly, when $\mathcal{I}\leq 30$. When $\mathcal{I}\geq 35$ the best performance was observed when the proposed algorithm was used. When the searching areas were increased to 10,  Figure \ref{fig:res5}, the EKF algorithm and the variants of the genetic algorithm have similar performance independently of the number of robots. But its performance is better than the particle filters algorithm when $\mathcal{I} \geq 20$. Finally when we consider the case with 20 search areas EKF fictitious play perform significantly better than both variants of the genetic algorithms we used. But its score is 3\% smaller than the one of particle filter alternative when we consider the case where 40 robots had to coordinate. 

We also test the computational cost, in seconds of the 50 iterations per agent of the algorithms we tested. For the genetic algorithms, since they are centralised algorithms, the cost per robot is the cost the algorithm needed to terminate. In order to have comparable results with the centralised cases, for the two distributed algorithms it is assumed that robots are able to communicate. Therefore, before they take an action there is a ``negotiation'' period among the robots. During this  period the robots exchange messages indicating the areas, which they intent to patrol. Then, when the final iteration of the negotiation period is reached, they execute the last joint action. Under this process the difference in the computational time of each robot is the total computational time of each algorithm. The experiments were performed on a PC with Intel i7 processor at 2.5 GHz processors and 8 GB memory. 

Even though the performance of the proposed algorithm is not always improving the results of the existing algorithms we tested, it has the smallest computational cost among them, as it is depicted in Tables \ref{tab:time1},\ref{tab:time2} and \ref{tab:time3}. This is important in robotics applications where restrictions, such as battery life, should be taken into account.
\begin{table}
\centering
\begin{tabular}{|c| c| c|c|c|}
\hline
 &One-shot Genetic algorithm &Genetic algorithm&Particle filter fictitious play &EKF fictitious play\\ \hline
 5&  1.097 & 56.9  & 0.185 & 0.159\\ \hline
 10& 2.081 & 105.9 & 0.399 & 0.165\\ \hline
 15& 2.255 & 110.6 & 0.616 & 0.246\\ \hline
 20& 2.358 & 117.5 & 0.886 & 0.326\\ \hline
 25& 2.366 & 121.5 & 1.089 & 0.428\\ \hline
 30& 2.698 & 159.9 & 1.374 & 0.503\\ \hline
 35& 2.728 & 132.1 & 1.558 & 0.589\\ \hline
 40& 2.839 & 136.7 & 1.917 & 0.660\\ \hline
 \end{tabular}
\caption{Average computational time that each robot was needed per replication, 5 areas }
 \label{tab:time1}
\end{table}
\begin{table}
\centering
\begin{tabular}{|c| c| c|c|c|}
\hline
 &One-shot Genetic algorithm &Genetic algorithm&Particle filter fictitious play &EKF fictitious play\\ \hline
 5 & 2.056& 106.9 & 0.628& 0.674\\ \hline
 10& 4.353& 190.5 & 1.165& 0.705\\ \hline
 15& 4.318& 199.85& 1.881& 1.105\\ \hline
 20& 3.296& 240.25& 2.579& 1.704\\ \hline
 25& 4.313& 246.75& 3.208& 2.119\\ \hline
 30& 5.092& 248.5 & 3.762& 3.143\\ \hline
 35& 5.275& 255.5 & 4.445& 3.899\\ \hline
 40& 6.023& 296.0 & 5.271& 4.487\\ \hline
 \end{tabular}
\caption{Average computational time that each robot was needed per replication, 10 areas }
 \label{tab:time2}
\end{table}
\begin{table}
\centering
\begin{tabular}{|c| c| c|c|c|}
\hline
 &One-shot Genetic algorithm &Genetic algorithm&Particle filter fictitious play &EKF fictitious play\\ \hline
 5&  3.455 & 136.9 & 0.826 & 0.706\\ \hline
 10& 12.023 & 556.9 & 1.906 & 1.105\\ \hline
 15& 20.141 & 749.2 & 2.932 & 1.704\\ \hline
 20& 20.681 & 870.8 & 3.276 & 2.119\\ \hline
 25& 22.014 & 875.1 & 4.955 & 2.229\\ \hline
 30& 22.211 & 909.8 & 6.138 & 3.143\\ \hline
 35& 26.146& 911.6 & 6.754 & 3.899\\ \hline
 40& 27.276 & 944.5 & 8.399 & 4.487\\ \hline
 \end{tabular}
\caption{Average computational time that each robot was needed per replication, 20 areas }
 \label{tab:time3}
\end{table}

\subsection{Transportation of a heavy object through corridors.}
In this section we describe how the proposed decision making module can be implemented as a coordination mechanism between two robots which carry a heavy item through a set of possible unmapped corridors. Based on a game-theoretic formulation, after the robots sense that an object or wall block their way, they have to coordinate in order to choose the action that will allow them to continue their mission even when they are not aware of the environment's structure. We assume that the two robots can move in  corridors using the moves that are depicted in Figure \ref{fig:moves}.  In order to take into account obstacles and varying widths of corridors we allowed two possible ways to turn left or right. In particular Action 1 represents the move forward, Action 2 represents the move forward for a distance $z$ meters and right $z$ meters, Action 3 is a $z$ meter diagonal right move of 45 degrees slope, Action 4 represents the move forward for $z$ meters and left $z$ meters and Action 5 is a $z$ meter diagonal left move of 45 degrees slope. This also introduces the need of a coordination mechanism because there can be cases, such as the set of corridors which are depicted in Figure \ref{fig:cor}, were both moves can be used by the robots to change direction. 

\begin{figure}
\centering
\includegraphics[scale=0.70]{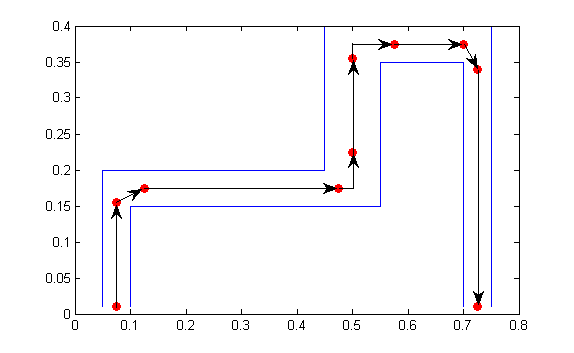}
\caption{The configuration the corridor used in our simulation in a MATLAB$^{TM}$ figure and illustration of robot moves.}
\label{fig:cor}
\end{figure}
\begin{figure}[h!]
\centering
\includegraphics[scale=0.70]{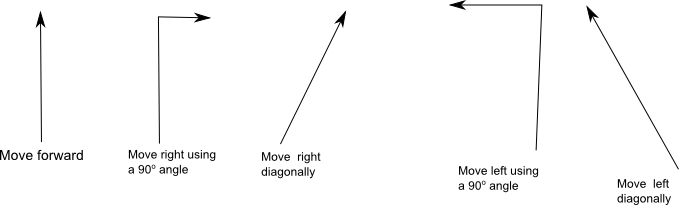}
\caption{Available moves to each robot}
\label{fig:moves}
\end{figure}

Robots in this scenario have to make decisions in various areas of Figure \ref{fig:cor}. Each of these areas is assumed to be a checkpoint, at which when the robots arrive, they use the EKF algorithm in order to choose the joint action that will lead them in the next checkpoint. The checkpoints can either be predefined before the task starts, or it can be dynamically generated based on set time intervals or the distance that the robots have covered.
In this game when the robots fail to coordinate, they receive zero utility, and when they coordinate the reward their gain is a constant $\mathfrak{c}$ if the joint action is feasible and zero otherwise. Thus in each checkpoint the robots play a game with rewards defined as:
\begin{equation}
r(s)= \left\lbrace
\begin{array}{ll}
\mathfrak{c}& \textrm{ if the joint action $s$ is feasible and $s^{1}=s^{2}$} \\
0 &\textrm{ otherwise}
\end{array}\right.
\end{equation}
A joint action $s$ is considered feasible when it is executed and the task can proceed, i.e. during the execution of the action neither the robots nor the  hit a wall.


We present the results over 100 replications of the above coordination task. In this task, because of the symmetry of the game, when the robots use  the particle filter variant of fictitious play, fail to coordinate the majority of the times. In particular only in 34\% of the replications of the coordination tasks were successfully completed when the particle filters algorithm was used. On the other hand, when the proposed algorithm was used, the task was completed in all replications.

Figure \ref{fig:dyncheck} displays a sample trajectory of the robots for the dynamically generated checkpoints. In this simulation scenario we assume that 0.1 distance units corresponds to 10 meters.  The checkpoints are generated when robots move 5 meters from the previous checkpoint, i.e. $z=5$. This can be seen also as a new checkpoint every 2 seconds when the robots are moving with speed of 2.5m/s.


\begin{figure}
\centering
\includegraphics[scale=0.70]{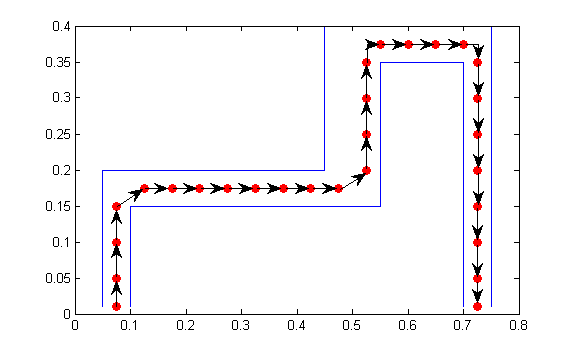}
\caption{Trajectory of the robots through the checkpoints.} 
\label{fig:dyncheck}
\end{figure}


Table \ref{tab:resdynam} shows the joint actions who can produce some reward (successful joint action) for each check point and the average number of iterations the EKF fictitious play algorithm needed to converge to a solution. 



\begin{table}
\centering
\begin{tabular}{|c| c| c|}
\hline
 &Successful joint actions&Iterations \\ \hline
 Checkpoint 1&(Action 1,Action 1)&1 \\ \hline
 Checkpoint 2&(Action 1,Action 1)&1 \\ \hline 
 Checkpoint 3&(Action 1,Action 1)&1  \\ \hline 
 Checkpoint 4&(Action 2,Action 2),(Action 3,Action 3)&25.6  \\ \hline 
 Checkpoint 5&(Action 1,Action 1)&1 \\ \hline
 Checkpoint 6&(Action 1,Action 1)&1 \\ \hline
 Checkpoint 7&(Action 1,Action 1)&1  \\ \hline
 Checkpoint 8&(Action 1,Action 1)&1  \\ \hline 
 Checkpoint 9&(Action 1,Action 1)&1  \\ \hline 
 Checkpoint 10&(Action 1,Action 1) &1  \\ \hline
 Checkpoint 11&(Action 1,Action 1)&1  \\ \hline
 Checkpoint 12&(Action 4,Action 4),(Action 5,Action 5)&29.8 \\ \hline
 Checkpoint 13&(Action 1,Action 1)&1 \\ \hline
 Checkpoint 14&(Action 1,Action 1) &1  \\ \hline 
 Checkpoint 15&(Action 1,Action 1)&1  \\ \hline 
 Checkpoint 16&(Action 1,Action 1)&1 \\ \hline
 Checkpoint 17&(Action 2,Action 2),(Action 3,Action 3)&22.6  \\ \hline
 Checkpoint 18&(Action 1,Action 1)&1  \\ \hline
 Checkpoint 19&(Action 1,Action 1)&1  \\ \hline
 Checkpoint 20&(Action 1,Action 1)&1  \\ \hline 
 Checkpoint 21&(Action 2,Action 2),(Action 3,Action 3)& 27.9 \\ \hline
 Checkpoint 22&(Action 1,Action 1)&1  \\ \hline 
 Checkpoint 23&(Action 1,Action 1)&1  \\ \hline
 Checkpoint 24&(Action 1,Action 1)&1  \\ \hline
 Checkpoint 25&(Action 1,Action 1)&1  \\ \hline
 Checkpoint 26&(Action 1,Action 1)&1  \\ \hline
 Checkpoint 27&(Action 1,Action 1)&1  \\ \hline
 \end{tabular}
\caption{Average number of iterations that was needed by EKF fictitious play algorithm in order for the robots to reach consensus.}
 \label{tab:resdynam}
\end{table}

\subsection{Ad-hoc sensor network}

The simulation scenario of this section is based on a coordination task of a power constrained sensor network, where sensors can be either in a sensinf or sleeping mode \citep{faraneli}. When the sensors are in a sensing mode, they can observe the events that occur in their range. During their sleeping mode the sensors harvest the energy they need in order to be able function when they are in the sensing mode. The sensors then should coordinate and choose their sensing/sleeping schedule in order to maximise the coverage of the events. This optimisation task can be cast as a potential game. In particular we consider the case where $\mathcal{I}$ sensors are deployed in an area where $E$ events occur. If an event $e$, $e \in E$, is observed from the sensors, then it will produce some utility $V_e$. Each of the sensors $i=1,\ldots,\mathcal{I}$ should choose an action $s^{i}=j$, from one of the $j=1,\ldots,J$ time intervals which they can be in sensing mode. Each sensor $i$, when it is in sensing mode, can observe an event $e$, if it is in its sensing range, with probability $p_{ie}=\frac{1}{d_{ie}}$, where $d_{ie}$ is the distance between the sensor $i$ and the event $e$. We assume that the probability that each sensor  observes an event is independent from the other sensors' probabilities. If we denote by $i_{in}$ the sensors that are in sense mode when the event $e$ occurs and $e$ is in their sensing range, then we can write the probability an event $e$ to be observed from the sensors, $i_{in}$ as:
\begin{equation}
1- \prod_{i \in i_{in}}{(1-p_{ie})}
\end{equation}
The expected utility that is produced from the event $e$ is the product of its utility $V_e$ and
the probability it has to be observed by the sensors, $i_{in}$ that are in sensing mode when the event $e$ occurs and $e$ is in their sensing range. More formally, we can express the utility that is produced from an event $e$ as:
\begin{equation}
r_{e}(s)=V_{e}(1- \prod_{i \in i_{in}}{(1-p_{ie}}))
\label{eq:target_utility}
\end{equation}
The global utility is then the sum of the utilities that all events, $e \in E$, produce
\begin{equation}
r_{global}(s)= \sum_{e}{r_{e}(s)}. \label{eq:global_utility}
\end{equation}
Each sensor, after each iteration of the game, receives some utility, which is based on the sensors and the events that are inside its communication and sensing range respectively. For a sensor $i$ we denote $\tilde{e}$ the events that are in its sensing range and $\tilde{s}^{-i}$ the joint action of the sensors that are inside his communication range. The utility that sensor $i$ will receive if his sense mode is $j$ will be 
\begin{equation}
r_{i}(s^{i}=j,\tilde{s}^{-i})= \sum_{\tilde{e}}{r_{\tilde{e}}(s^{i}=j,\tilde{s}^{-i})} \label{eq:individual_utility}.
\end{equation}

In this simulation scenario 40 sensors are placed on a unite square, each have to choose one time interval of the day when they will be in sensing mode and use the rest time intervals to harvest energy. We consider cases where sensors had to choose their sensing mode among 4 available time intervals. Sensors are able to communicate with other sensors that are at most 0.6 distance units away, and can only observe events that are at most 0.3 distance units away. The case where 20 events are taking place is considered. Those events were uniformly distributed in space and time, so an event could evenly appear in any point of the unit square area and it could occur at any time with the same probability. The duration of each event was uniformly chosen between $(0-6]$ hours and each event had a value $V_{e} \in (0-1]$. 

The results of EKF fictitious play are compared with the ones of the particle filter based algorithm. The number of particles in particle filter fictitious play algorithm, play an important role in the computational time of the algorithm. For that reason comparisons were made with three variants of the particle filter algorithm with 500-1000 particles.

The results presented in Figure \ref{fig:sens} and Table \ref {tab:sens} are the averaged reward of 100 instances of the above described scenario. Figure \ref{fig:sens} depicts the the average global reward. The global reward in each instance was normalised using the maximum reward which is achieved when all the sensors are in sense mode. As it is depicted in Figure  \ref{fig:sens} the average performance of the particle filter algorithm is not affected by the number of particles in this particular example. But the EKF fictitious play performed better than the particle filters alternatives. In addition, as it is shown in Table \ref{tab:sens}, the computational cost of the two particle filter's variants is greater than the EKF fictitious play algorithm. 

\begin{figure}
\centering
\includegraphics[scale=0.65]{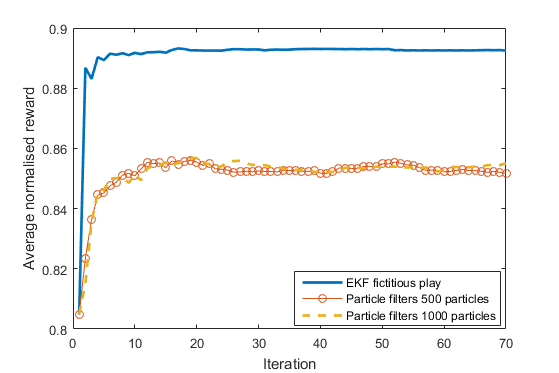}
\caption{Average reward as a function of time of the three learning algorithms. The x-axis represents the Iteration and the y-axis the average normalised reward.}
\label{fig:sens}
\end{figure}

\begin{table}
\centering 
\begin{tabular}{|c|c|c|c|}
\hline
&EKF fictitious play& Particle filter fictitious play 500 particles &Particle filter 1000 particles \\
\hline
Average time& 0.33&5.37 &8.43 \\
\hline
\end{tabular}
\caption{Average computational time for each agent in seconds.}
\label{tab:sens}
\end{table}

\section{Conclusions}
A novel variant of fictitious play has been presented. This variant is based on extended Kalman filters and can be used as a decision making module of a robot's controller. The relation between the proposed learning process and the BDI framework was also considered. The main theoretical result has been,  that the learning algorithm associated with the cooperative game converges to the Nash equilibrium of potential games.  The application of this methodology has been  illustrated on four examples.  The filtering methods used for learning have also been compared  for their performance and computational costs.  The methods presented can potentially find widespread applications in the robotics industry. Possible improvements could be obtained by the adoption of more sophisticated reward/cost functions without altering the essential features of the methodology. 

\appendix
\section{Proof of Proposition 1}
\label{append1}

\subsection{Case where players have two available actions}
In the case where players have only 2 available actions Player $i$'s predictions about his opponent's propensity are:

\begin{equation}
\bar{Q}_{t}= \left(
\begin{array}{c}
\hat{Q}_{t-1}[1] \\
\hat{Q}_{t-1}[2] \\
\end{array} \right)
\end{equation}

\begin{equation}
P_{t}^{-} = \left(
\begin{array}{cc}
P^{-}_{t-1}[1,1] &P^{-}_{t-1}[1,2] \\
P^{-}_{t-1}[2,1]& P^{-}_{t-1}[2,2]  \\
\end{array} \right)+ qI
\end{equation}
without loss of generality we can assume that his opponent in iteration $t$ chooses action 2. Then the update step will be :
\begin{equation}
 v_t=z_{t}-h(\bar{Q}_{t})
\end{equation}
since Players $i$'s opponent played action 2 and $h_{s(k)}(Q)=\frac{exp(Q_{t}[s(k)]/\tau)}{\sum_{\tilde{s} \in S} exp(Q_{t}[\tilde{s}]/\tau)}$ we can write $v_{t}$ and $H_{t}(\bar{Q}_t)$ as:
\begin{eqnarray}
v_{t} = &\left(
\begin{array}{c}
0 \\
1 \\
\end{array} \right)-
\left(
\begin{array}{c}
\sigma_{t-1}(1) \\
1- \sigma_{t-1}(1) \\
\end{array} \right) \nonumber \\
 = &\left(
\begin{array}{c}
- \sigma_{t-1}(1) \\
\sigma_{t-1}(1) \\
\end{array}  \nonumber \right)
\end{eqnarray}

\begin{equation}
H_{t}(\bar{Q}_{t}) = \left(
\begin{array}{cc}
a_{t} & -a_{t} \\
-a_{t}& a_{t}  \\
\end{array} \right)
\end{equation}
where $a_{t}$ is defined $a_{t}=\sigma_{t-1}(1)\sigma_{t-1}(2)$. The estimation of \mbox{$S_{t}=H(\bar{Q}_{t})P_{t}^{-}H^{T}(\bar{Q}_{t})+Z$} will be: 
\begin{equation}
S_{t} = a^{2}\left(
\begin{array}{cc}
b & -b \\
-b& b 
\end{array} \right)+Z
\end{equation}
where $b=P^{-}_{t}[1,1]+P^{-}_{t}[2,2]-2P^{-}_{t}[1,2]$. The Kalaman gain, $K_{t}=P_{t}^{-}H^{T}(\bar{Q}_{t})S_{t}^{-1}$ can be written as 

\begin{equation}
K_{t}= \frac{1}{2r b+r^2}\left(
\begin{array}{cc}
P_{t}^{-}[1,1] & k \\
k& P_{t}^{-}[2,2]  \\

\end{array} \right) \left(
\begin{array}{cc}
a_{t}& -a_{t} \\
-a_{t}&a_{t}  \\

\end{array} \right)\left(
\begin{array}{cc}
b+r & b\\
b& b+r \\
\end{array} \right)
\end{equation}
up to a multiplicative constant we can write
\begin{equation}
K_{1} \sim \left(
\begin{array}{cc}
c & -c\\
-d& d \\

\end{array} \right)
\end{equation}
where $c=P_{t}^{-}[1,1]-P_{t}^{-}[1,2]$ and $d=P_{t}^{-}[2,2]-P_{t}^{-}[1,2]$.  The EKF estimate of the opponent's propensity is:  
\begin{eqnarray}
\label{eq:mupd}
\hat{Q}_{t} = &\left(
\begin{array}{c}
\hat{Q}_{t}[1] \\
\hat{Q}_{t}[2] \\
\end{array} \right) 
 =\left(
\begin{array}{c}
\bar{Q}_{t}[1]-2\sigma(1)\frac{a(b-k)}{4a^{2}(b-k)+(r+\epsilon)}\\
\bar{Q}_{t}[2]+2\sigma(1)\frac{a(b-k)}{4a^{2}(b-k)+(r+\epsilon)}\\
\end{array} \right)
\end{eqnarray}
Based on the above we observe that $\hat{Q}_{t}(1)<\hat{Q}_{t-1}(1)$ and $\hat{Q}_{t}(2)>\hat{Q}_{t-1}(2)$.

\subsection{Case where players have more than two available actions}

In the case where there are more than 2 available actions, when $1/t<<1$ then the covariance matrix $S_{t} \simeq H(\bar{Q}_{t})P_{t}^{-}H^{T}(\bar{Q}_{t})$ and then $K_{t} \simeq H(\bar{Q}_{t})$. From the definition of $H(\bar{Q}_{t}$ we know that it is invertible, with positive diagonal elements  and negative off diagonal elements. In particular we can write :

\begin{equation*}
H(\bar{Q}_{t})[i,i]=\sum_{j \in S /i} \sigma_{t-1}(s^{i})\sigma_{t-1}(s^{j})
\label{eq:hii}
\end{equation*}

\begin{equation*}
H(\bar{Q}_{t})[i,j]=-\sigma_{t-1}(s^{i})\sigma_{t-1}(s^{j}).
\label{eq:hij}
\end{equation*}

Suppose that action $s^{j}$ is played from Player $i$'s opponent then the update of $\hat{Q}_{t}[j]$ is the following:

\begin{equation*}
\hat{Q}_{t}[j]=\bar{Q}_{t}+H(\bar{Q}_{t})[j,\cdot]y.
\label{eq:mtup}
\end{equation*}
The only positive element of $y$ is $y[j]$. The multiplication of $H(\bar{Q}_{t})[j,\cdot]$ and $y$ is the sum of $\mathbb{I}$ positive coefficients and therefore the value of $\hat{Q}_{t}[j]$ will be increased. In order to complete the proof we should show that $\Delta[j] > \Delta[i] \forall i \in S /i$, where $\Delta[\tilde{i}]=\hat{Q}_{t}[\tilde{i}]-\hat{Q}_{t-1}[\tilde{i}], \tilde{i} \in S$.  For simplicity of notation for the rest of the proof we will write $H[i,j]$ instead of $H(\bar{Q}_{t})[i,j]$ and $\sigma(i)$ instead of $\sigma_{t-1}(s^{i})$. 

\begin{equation}
\begin{array}{r l}
\Delta[i]=& H[i,\cdot]y \\
=&H[i,1]\sigma(1)+H[i,2]\sigma(2)+\ldots+\\
&H[i,i-1]\sigma(i-1)-H[i,i]\sigma(i)+\\
&H[i,i+1]\sigma[i+1]+\ldots+\\
&H[i,j-1]\sigma(j-1)-H[i,j](1-\sigma(j))+\\
&H[i,j+1]\sigma(j+1)+\ldots+H[i,\mathcal{I}]\sigma(\mathcal{I})
\end{array}
\label{eq:deltai}
\end{equation}

\begin{equation}
\begin{array}{r l}
\Delta[j]=& H[j,\cdot]y \\
=&H[j,1]\sigma(1)+H[j,2]\sigma(2)+\ldots+\\
&H[j,j](1-\sigma(j))+\ldots+H[i,\mathcal{I}]\sigma(\mathcal{I})
\end{array}
\label{eq:deltaj}
\end{equation}

\begin{equation}
\begin{array}{r l}
\Delta[i]-\Delta[j]=& \sigma(1)(H[i,1]-H[j,1])+ \ldots+\\
&\sigma(i-1)(H[i,i-1]-H[j,i-1])-\\
&\sigma(i)(H[i,i]+H[j,i])+\\
&\sigma(i+1)(H[i,i+1]-H[i,i+1])+ \ldots+\\
&\sigma(j-1)(H[i,j-1]-H[j,j-1]) -\\
&(1-\sigma(j))(H[[i,j]+H[j,j]])+\\
&\sigma(j+1)(H[i,j+1]-H[j,j+1]) +\ldots+\\
&\sigma(\mathcal{I})(H[i,\mathcal{I}]-H[j,\mathcal{I}])	
\end{array}
\label{eq:difdel}
\end{equation}
If we substitute $H[\cdot,\cdot]$ with its equivalent and we can write $\Delta[i]-\Delta[j]$ as:
\begin{equation}
\begin{array}{r l}
\Delta[i]-\Delta[j]=& (\sigma(1))^{2}(\sigma(i)-\sigma(j))+ \ldots+\\
&(\sigma(i-1))^{2}(\sigma(i)-\sigma(j))-\\
&(\sigma(i))^2((\sum_{\tilde{j} \in S /i} \sigma(\tilde{j}))+\sigma(j))+\\
&(\sigma(i+1))^2(\sigma(i)-\sigma(j))+ \ldots+\\
&(\sigma(j-1))^2(\sigma(i)-\sigma(j))-\\
&(1-\sigma(j))\sigma(j)((\sum_{\tilde{j} \in S /j} \sigma(\tilde{j}))+\sigma(i))+\\
&(\sigma(j+1))^{2}(\sigma(i)-\sigma(j)) +\ldots+\\
&(\sigma(\mathcal{I}))^{2}(\sigma(i)-\sigma(j))	
\end{array}
\label{eq:difdel1}
\end{equation}
solving the inequality $\Delta[i]-\Delta[j]<0$ we obtain:
\begin{equation}
\begin{array}{l}
\Delta[i]-\Delta[j]<0 \Leftrightarrow\\
(\sigma(i)-\sigma(j))((\sum_{\tilde{j} \in S /\{i,j\}} (\sigma(\tilde{j})^{2})))<\\
((\sigma(i))^{2}((\sum_{\tilde{j} \in S /i} \sigma(\tilde{j}))+\sigma(j))+\\
(1-\sigma(j))(\sigma(j))((\sum_{\tilde{j} \in S /j} \sigma(\tilde{j}))+\sigma(i))	)
\end{array}
\label{eq:ineq1}
\end{equation}
In the case where $\sigma(i)<\sigma(j)$ the inequality is satisfied always because the left hand side of the inequality is always negative and the right hand side is always positive. In the case where $\sigma(i)>\sigma(j)$ inequality (\ref{eq:ineq1}) we will show by contradiction that (\ref{eq:ineq1}) is satisfied $\forall i, i \neq j$. Therefore we assume that:
\begin{equation}
\begin{array}{l}
\Delta[i]-\Delta[j]\geq 0 \Leftrightarrow\\
(\sigma(i)-\sigma(j))((\sum_{\tilde{j} \in S /\{i,j\}} (\sigma(\tilde{j})^{2})))\geq\\
((\sigma(i))^{2}+(1-\sigma(j))(\sigma(j)))((\sum_{\tilde{j} \in S /i} \sigma(\tilde{j}))+\sigma(j))	
\end{array}
\label{eq:ineq3}
\end{equation}
Since $((\sum_{\tilde{j} \in S /\{i,j\}} (\sigma(\tilde{j})^{2})))<((\sum_{\tilde{j} \in S /i} \sigma(\tilde{j}))+\sigma(j))$ in order to complete the proof we only need to show that:

\begin{equation}
\begin{array}{r l}
\sigma(i)-\sigma(j)\geq&(\sigma(i))^{2}+(1-\sigma(j))\sigma(j)	\\
\sigma(i)-(\sigma(i))^{2}\geq&2\sigma(j)-(\sigma(j))^2 \\
\sigma(i)(1-\sigma(i))\geq&\sigma(j)(2-\sigma(j))
\end{array}
\label{eq:ineq4}
\end{equation}
Inequality (\ref{eq:ineq4}) will be satisfied if the following inequality is satisfied:
\begin{equation}
\begin{array}{r l}
(1-\sigma(i))\geq& \frac{\sigma(j)}{\sigma(i)}(2-\sigma(j))
\end{array}
\label{eq:ineq5}
\end{equation}
Inequality (\ref{eq:ineq5}) will be satisfied if 
\begin{equation}
\begin{array}{r l}
(1-\sigma(i))\geq& (2-\sigma(j))\\
\sigma(j)-\sigma(i)\geq1
\end{array}
\label{eq:ineq6}
\end{equation}
since $\sigma(i)>\sigma(j)$, (\ref{eq:ineq6}), is false $\forall i$ and thus by contradiction $\Delta[i]-\Delta[j] <0$, which completes the proof.

\section{Proof of Proposition 2}
\label{append2}
Similarly to the proof of Proposition \ref{prop1} we consider only one opponent and a game with more than one pure Nash equilibria. We assume that a joint action $s=(s^{1}(j'),s^{2}(j))$ is played which is not a Nash equilibrium. Since players use EKF fictitious play because of Proposition \ref{prop1} they will eventually change their action to $\tilde{s}=(s^{1}(\tilde{j'}]), s^{2}(\tilde{j}))$ which will be the  best response to action $s^{2}(j)$ and $s^{1}(j')$ for players 1 and 2 respectively. But if this change is simultaneous their is no guarantee that the resulted joint action $\tilde{s}$ will increase the expected reward of the players. We will show that with high probability the two players will not change actions simultaneously infinitely often. 

Without loss of generality we assume that at time $t$ Player 2 change his action from $s^{2}(j)$ to $s^{2}(\tilde{j})$. We want to show that the probability that Player 1 will change his action to $s^{1}(\tilde{j'})$ with probability less than 1. Player 1 will change his action if he thinks that Player 2 will play action $s^2(\tilde{j})$ with high probability such that his utility will be maximised when he plays action $s^{1}(\tilde{j'})$. Therefore Player 1 will change his action to $s^{1}(\tilde{j'})$ if $\sigma^{2}(s^{2}(j))>\lambda$, where $\sigma^{2}(\cdot)$ is the estimation of his Player 1 about Player 2's strategy. Thus we want to show 

\begin{equation}
\begin{array}{l}
p(\sigma^{2}>\lambda)<1 \Leftrightarrow \\
p(\frac{\exp(\hat{Q}_{t-1}[j])}{\sum_{j'' \in S^{2}}{\exp(\hat{Q}_{t-1}[j''])}}>\lambda)<1 \Leftrightarrow \\
p(\hat{Q}_{t-1}[j]>\ln \frac{\lambda}{1-\lambda} + \sum_{j'' \in S^{2}/j}{\exp(\hat{Q}_{t-1}[j''])})<1 \Leftrightarrow \\
p(\hat{Q}_{t-2}[j]+(K_{t-1}y_{t-1})[j]>\ldots\\
\quad \ln \frac{\lambda}{1-\lambda} + \sum_{j'' \in S^{2}/j}{\exp(\hat{Q}_{t-1}[j''])})<1 
\end{array}
\label{eq:prop2in}
\end{equation}
we can expand $K_{t-1}y_{t-1}[j]$ as follows
\begin{equation}
\begin{array}{r l}
(K_{t-1}y_{t-1})[j]=&((P_{t-2}+(\tilde{\xi}+\epsilon)I)HS^{-1}y_{t-1})[j] \\
=&(P_{t-2}HS^{-1}y_{t-1})[j]+((\tilde{\xi}+\epsilon)HS^{-1}y_{t-1})[j] 
\end{array}
\label{eq:kexpand}
\end{equation}
If we substitute $K_{t-1}y_{t-1}[j]$ in (\ref{eq:prop2in}) with its equivalent we obtain the following inequality:

\begin{equation}
\begin{array}{l}
p(\hat{Q}_{t-2}[j]+(P_{t-2}HS^{-1}y_{t-1})[j]+((\tilde{\xi}+\epsilon)HS^{-1}y_{t-1})[j]>\ldots\\
\quad \ln \frac{\lambda}{1-\lambda} + \sum_{j'' \in S^{2}/j}{\exp(\hat{Q}_{t-1}[j''])})<1 \Leftrightarrow \\
p(((\tilde{\xi}+\epsilon)HS^{-1}y_{t-1})[j]> -\hat{Q}_{t-2}[j]-(P_{t-2}HS^{-1}y_{t-1})[j] \ldots \\
\quad \ln \frac{\lambda}{1-\lambda} + \sum_{j'' \in S^{2}/j}{\exp(\hat{Q}_{t-1}[j''])})<1 \Leftrightarrow \\
p(\epsilon>C)<1
\end{array}
\label{eq:prop2fin}
\end{equation}
where $C$ is defined as: 
 
\begin{equation}
\begin{array}{r l}
C=&\frac{\ln \frac{\lambda}{1-\lambda}}{(HS^{-1}y_{t-1})[j]}+\\ 
&\frac{\sum_{j'' \in S^{2}/j}{\exp(\hat{Q}_{t-1}[j''])}}{(HS^{-1}y_{t-1})[j]}-\\
&\frac{-\hat{Q}_{t-2}[j]-(P_{t-2}HS^{-1}y_{t-1})[j]}{(HS^{-1}y_{t-1})[j]}-\\
&\frac{q(HS^{-1}y_{t-1})[j]}{(HS^{-1}y_{t-1})[j]}
\end{array}
\label{eq:prop2C}
\end{equation}

Inequality (\ref{eq:prop2fin}) is always satisfied since $\epsilon$ is a Gaussian random variable. To conclude the proof we define $\chi_{t}$ as the event that both players change their action at time $t$ simultaneously, and assume that the two players have change their actions simultaneously at the following iterations $t_{1}, t_{2}, \ldots, t_{T}$, then the probability that they will also change their action simultaneously at time $t_{T+1}$, $P(\chi_{t_1},\chi_{t_2}, \ldots, \chi_{t_T}, \chi_{t_T+1})$ is almost zero for large but finite $T$.


\bibliographystyle{elsarticle-harv}
\bibliography{ekf,IEEEabrv,bibfile}







\end{document}